\documentclass[reprint,aps,pra,groupedaddress]{revtex4-1}

\usepackage{amsmath, amsthm}
\usepackage{mathtools}
\usepackage{amsfonts}
\usepackage{xcolor}
\usepackage{qcircuit}

\usepackage{subcaption}

\makeatletter
\let\@makecaption=\caption@makecaption
\makeatother

\usepackage[colorlinks=true,citecolor=blue,linkcolor=red,urlcolor=blue]{hyperref}

\newtheorem{theorem}{Theorem}
\newtheorem{lemma}{Lemma}
\newcommand{\be}{\begin{equation}}
\newcommand{\ee}{\end{equation}}
\newcommand{\ben}{\begin{eqnarray}}
\newcommand{\een}{\end{eqnarray}}
\newcommand{\bes}{\begin{subequations}}
\newcommand{\ees}{\end{subequations}}
\newcommand{\bF}{\begin{figure}}
\newcommand{\eF}{\end{figure}}

\def\ket#1{| #1 \rangle}

\def\bra#1{\langle #1 |}

\usepackage{enumerate}

\begin{document}

\title{Fault-tolerant quantum metrology}

\author{Theodoros Kapourniotis}

\affiliation{Department of Physics, University of Warwick, Coventry CV4 7AL, United Kingdom}

\author{Animesh Datta}
\affiliation{Department of Physics, University of Warwick, Coventry CV4 7AL, United Kingdom}
	
\date{\today}

\begin{abstract}
We introduce the notion of fault-tolerant quantum metrology to overcome noise beyond our control -- associated with sensing the parameter, by reducing the noise in operations under our control -- associated with preparing and measuring probes and ancillae. 
To that end, we introduce noise thresholds to quantify the noise resilience of parameter estimation schemes.
We demonstrate improved noise thresholds over the non-fault-tolerant schemes.
We use quantum Reed-Muller codes to retrieve more information about a single phase parameter being estimated in the presence of full-rank Pauli noise. Using only error detection, as opposed to error correction, allows us to retrieve higher thresholds. We show that better devices, which can be engineered, can enable us to counter larger noise in the field beyond our control.
Further improvements in fault-tolerant quantum metrology could be achieved by optimising in tandem parameter-specific estimation schemes and transversal quantum error correcting codes.

\end{abstract}


\maketitle

\section{Introduction}
 Like all quantum information processing tasks, noise has an adverse effect on quantum enhancements in precision metrology. 
Early promises of a quantum-enhanced `Heisenberg scaling' are now tempered by its elusiveness even in the presence of arbitrarily small noise in the sensing process~\cite{demkowicz2012elusive,Smirne2016}. 
After some early musings~\cite{Preskill2000,Macchiavello2000}, much effort has been directed towards recovering the `Heisenberg scaling' using quantum error correction~\cite{Ozeri2013,dur2014improved,Kessler2014,Arrad2014,HerreraMarti2015,Lu2015,Unden2016}, 
More recent results suggest the impossibility of recovering the `Heisenberg scaling' in the presence of general Markovian noise if the Hamiltonian lies in the span of the noise operators, even after quantum error correction~\cite{Demkowicz2017,zhou2017achieving}. 
Studies of error-corrected quantum metrology have either focussed on specific experimental systems~\cite{HerreraMarti2015,Unden2016,Reiter2017,Layden17} or specific forms of noise affecting the field~\cite{dur2014improved,HerreraMarti2015}. Others have assumed instantaneous and perfect correction and control operations~\cite{Arrad2014,HerreraMarti2015,Demkowicz2017,zhou2017achieving} or short sensing times to commute noise to the end of the protocol~\cite{dur2014improved,Kessler2014,Lu2015}. These assumptions are unlikely to hold in general.

In this article, we take a complementary approach by initiating the study of fault-tolerant (FT) quantum metrology. 
Instead of lower bounds and asymptotic scalings, we focus on the estimation of a phase parameter $\phi$ associated with the field
\be
R_z(\phi) = \exp \left(-i \frac{\phi}{2}Z\right),
\label{eq:Rz}
\ee
up to a fixed number of bits, where $Z=\ket{0}\bra{0}-\ket{1}\bra{1}$. 
We show that $\phi$ can be estimated to more bits of precision with our FT quantum metrology protocol, in the presence of noise, than without it.
This is achieved by introducing the concept of \emph{thresholds} to noisy quantum metrology, providing experimentalists with quantitative targets to aim for.
Our illustration uses a specific phase estimation scheme and code switching between Steane and other quantum Reed-Muller codes (QRMCs) to counter locally bounded full-rank noise \emph{beyond} our control -- associated with the parameter or field being sensed, as well as \emph{under} our control -- in preparing, entangling, and measuring probes and ancilla. We call the latter `devices'. We do not assume 
short sensing times or
 perfect control operations. 
We show in Fig.~(\ref{fig:threshold_relations}) that better devices, which can be engineered, can enable us to counter more noise in the field beyond our control. 
Our results for fault tolerant quantum metrology can also be extended to other sensing and estimating applications, such as clock synchronisation~\cite{de2005quantum}, and systematic error estimation and calibration~\cite{kimmel2015robust}.

In contrast to previous approaches of error-corrected quantum metrology such as~\cite{dur2014improved,Demkowicz2017,zhou2017achieving,Layden2019}, as well as
ancilla-assissted quantum metrology schemes~\cite{Demkowicz2014} our FT quantum metrology framework enables a meaningful quantitive analysis of the noise in devices in addition to that in the field.
Recast in terms of noise thresholds, these prior works on quantum metrology with prefect error-correction corresponds to the blue solid line in Fig.~(\ref{fig:threshold_relations}). The blue dashed line shows the depreciating performance of an error-corrected quantum metrology scheme due to noisy devices. Our main result is the green line in Fig.~(\ref{fig:threshold_relations}) that shows the possibility of improvement using fault tolerant quantum metrology.

This paper is organised as follows. In Sec.~(\ref{sec:ftmet}), we introduce the notion of fault tolerance in quantum metrology, comparing and contrasting it with the more familiar notion of fault tolerance in quantum computing.  Sec.~(\ref{Sec:res}) presents out main results, culminating in Fig.~(\ref{fig:threshold_relations}). The subsequent sections provide the technical details and proofs. Sec.~(\ref{AppendixRG}) provides formal convergence, noise resilience and resource analysis of a modified estimation scheme~\cite{rudolph2003quantum}. Sec.~(\ref{Appendix_disc}) calculates the effect of applying logical $R_z(\phi)$, which is non-trasversal for QRMCs in general. Sec.~(\ref{Appendix_thres}) calculates failure probabilities of error detection when devices are perfect. Sec.~(\ref{Appendix_noisydev}) analyses the performance of the protocols when devices are not perfect. Finally, Sec.~(\ref{appendix_parallel}) presents the parallel version of our protocols and Sec.~(\ref{sec:disc}) discusses prospects and open questions in fault tolerant quantum metrology.

\section{Fault-tolerance and metrology}
\label{sec:ftmet}

We treat phase estimation as a quantum circuit, composed of the probe and ancilla state preparations and measurements as well as the application of $R_z(\phi)$ gate. The central difference between FT quantum metrology and computing is that $\phi$ is unknown in the former while it is known in the latter. The only way to apply $R_z(\phi)$ is by interrogating the field. 

Quantum information can be protected against bounded noise by using a quantum error correcting code (QECC). In order to protect it while it dynamically undergoes computation one can apply the procedures of fault tolerance.  Fault tolerance encompasses a set of procedures for preparing encoded states, applying encoded gates, and measuring encoded states.  If $\phi$ is known, as is the case in computing, a fault tolerant encoding of $R_z(\phi)$ in Eqn.~(\ref{eq:Rz}) can be accomplished. This relies on the existence of a fault tolerant set of gates from which to build a fault tolerant circuit. The main property of these procedures is that an error in one component in a FT encoding results in no more than one error in the entire encoded block~\cite{nielsen2010quantum}. As $\phi$ is unknown in metrology, we cannot undertake its fault tolerant encoding directly. Fault tolerant quantum metrology thus operates by performing fault tolerance before and after the field $R_z(\phi)$ is sensed, as in Fig.~(\ref{fig:protocol2}).

\begin{figure}[t]	
	\mbox{
\Qcircuit  @C=0.8em @R=1em {
		 &   \qw	& \gate{U}		& \qw	  	 \\
		 &   \qw	& \gate{U}		& \qw	  	 \\
	 \lstick{\raisebox{2.6em}{$\ket{\psi}_L$\ }}	 &  & 	\cdots	& 	  	 \\
			&   \qw	& \gate{U}		& \qw	 \inputgroupv{1}{4}{0.8em}{2.6em}{}
			  }
			}
		\caption{Single-qubit state $\ket{\psi}$ is encoded into multi-qubit state $\ket{\psi}_L$ in order to detect/correct errors on few qubits. Transversal application of unitary gate $U$ means bitwise application of $U$ on physical qubits of $\ket{\psi}_L$. It outputs the encoded state of state $U \ket{\psi}$.
		} \label{transversality}
\end{figure}
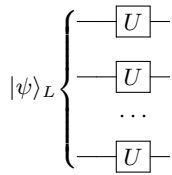

A desirable design principle in fault tolerance is limiting the proliferation of noise from one part of the circuit to another. This is called transversality and is the requirement that each physical gate employed for the encoded gate acts on at most one physical qubit in each code block~\cite{Raussendorf2012}, as shown in Fig.~(\ref{transversality}). 
Since in FT quantum metrology only single qubit gates $R_z(\phi)$ are applied during the interrogation of the field, as shown in Fig.~(\ref{fig:protocols}), transversality comes naturally. 
It results in errors on a single physical qubit not propagating to more physical qubits of the same block in a single fault-tolerant gate procedure. 

If we restrict ourselves to the well-studied family of stabilizer codes, we cannot hope for a code transversal for $R_z(\phi), \forall \phi \in [0,2\pi]$. This is because for stabilizer codes all transversal gates reside at a finite level of the Clifford hierarchy~\cite{jochym2017disjointness} (for details see Appendix~\ref{appendix_ft}). We must therefore move to a digital representation of the phase parameter $\phi=2 \pi \times 0.b_0 b_1 b_2 \ldots=b_0 \pi + b_1 \pi/2 + b_2 \pi/4 + \ldots$ with $b_n \in \{0,1\}$.  Defining $T_{n} \equiv \text{diag}(1, e^{i2\pi /2^{n}}),$ Eq.~(\ref{eq:Rz}) can be re-expressed as $R_z(\phi)= T_1^{b_0} T_2^{b_1} \ldots$. Thus, the field interrogation effectively does or does not apply the gate $T_n$ depending on whether $b_n =0$ or $b_n =1 $ respectively.  For $n$ higher than what our transversal code can support, there is a corruption of the logical subspace. We prove in Sec.~(\ref{Appendix_disc}) that this effect is  bounded and using stabilizer codes can even be beneficial in our construction. 
Since any real-world task must use finite resources, we capture the performance of FT quantum metrology in the number of bits of $\phi$ estimated.
Incidentally, digital quantum metrology has been studied for independent reasons~\cite{Hassani2017}.

Other design principles of fault tolerance quantum computing include gate synthesis/approximation to acquire a FT gate set~\cite{Dawson06}, distillation of so-called magic states~\cite{bravyi2005universal}, and state twirling to diagonalise the noise in the magic state basis~\cite{campbell2010bound}. In the following, we briefly describe why these cannot be applied to FT quantum metrology in their original form and the modifications we resort to. 

Gate synthesis replaces gates that do not belong to the FT set by approximate decompositions of gates of in that set. This cannot be applied in FT quantum metrology since we cannot write a decomposition of the gate $R_z(\phi)= T_1^{b_0} T_2^{b_1} \ldots$ when the bits $b_n$ are unknown. The only way to apply the gate is by interrogating the field. This results in using larger block size QECC for retrieving more bits of the unknown parameter in our FT quantum metrology scheme.

The gates involved in the encoding operations (Hadamard and controlled-NOT) and the field ($R_z(\phi)= T_1^{b_0} T_2^{b_1} \ldots$) form a gate set universal for quantum computing.
The well-studied family of stabilizer codes is known not to be transversal for a universal set of gates~\cite{eastin2009restrictions}. A solution is to inject external states into the logical circuit in order to apply the corresponding gates. Distillation is a series of operations that gives a high fidelity state out of many low fidelity states and is necessary because the external state is noisy in general. It is accompanied by gate teleportation to apply the corresponding gate at any stage of the circuit, as shown in Fig.~(\ref{teleport1}). This circuit cannot be implemented in our FT quantum metrology scheme, once again because of the unknown $\phi$-dependent correction operator in the teleportation step. An alternative solution, which avoids teleportation, is code switching. It switches between codes that are transversal for different subsets of gates. This solution we use (Sec.~(\ref{Appendix_why})), which has implications on the noise threshold of our scheme.

\begin{figure}[t!]
\mbox{
\Qcircuit  @C=1.6em @R=1.0em {
		& \lstick{\ket{\psi}}& \qw	& \targ		& \qw		& \measureD{Z}		& \cw		& \control \cw	 \\
			& \lstick{\ket{+}}& \gate{R_z(\phi)}		&\ctrl{-1}	& \qw		& \qw		& \qw &  \gate{R_z(2\phi) X}  \cwx  \gategroup{2}{1}{2}{3}{.7em}{--} }
			}
\caption{Gate teleportation: All operations outside the dashed box are protected by a code transversal for $\{cX,H,Z\}$. The unitary correction depends on parameter $\phi$ and since it is not transversal for the code requires an extra round of distillation.}
\label{teleport1}
\end{figure}
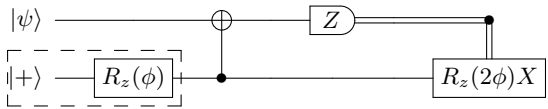

State twirling is the application of a randomising operation that diagonalises the noise on a state in a basis that is defined by the state. It reduces the types of logical noise that need to be treated in a FT circuit. In our FT quantum metrology scheme, this also cannot be applied because after the first interrogation the state depends on the unknown parameter $\phi$. FT quantum metrology thus needs to treat full rank noise in its entirety.

The culmination of a fault tolerant approach is the threshold. The noise threshold for FT quantum metrology we define as the strength of noise~\footnote{Noise strength is typically defined as the diamond norm of the noise operators} below which the estimator for the parameter converges.
It depends on the type of noise, the estimation scheme, and the QECC used. 
Indeed, our FT metrology scheme provides two thresholds - one for the noise in the field which contains the parameter, and another for the devices that perform the preparation, encoding, and measurements of the probes and ancilla. 
One of our main results as shown in Fig.~(\ref{fig:threshold_relations}) is that if the noise in the devices in below a certain threshold, then the threshold for the noise in the field is larger.

In FT quantum computing using gate distillation, a base code transversal only for Clifford gates with high noise threshold such as the surface code~\cite{raussendorf2007topological} can be used. Furthermore, the distillation is based on error detection rather than error correction which contributes to a higher threshold. 
In FT quantum metrology where we have to use code switching, estimating bit $b_n$ requires us to employ a code that is transversal for $T_n.$
In this work, we use Quantum Reed-Muller codes (QRMCs). 
The search for QECCs with better performance, in terms of codelengths and thresholds should be one of the central aims of improving FT quantum metrology in future works.

\begin{figure*}
	\begin{subfigure}[t]{2.2in}
\includegraphics[scale=0.7]{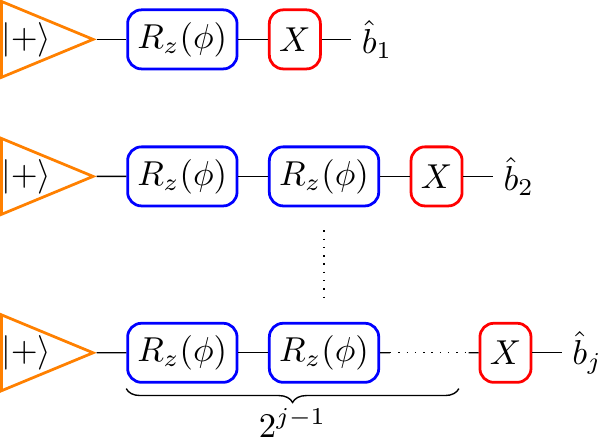}
		\caption{Protocol Ia}\label{fig:protocol1}		
	\end{subfigure}
	\begin{subfigure}[t]{2.2in}
		\includegraphics[scale=0.7]{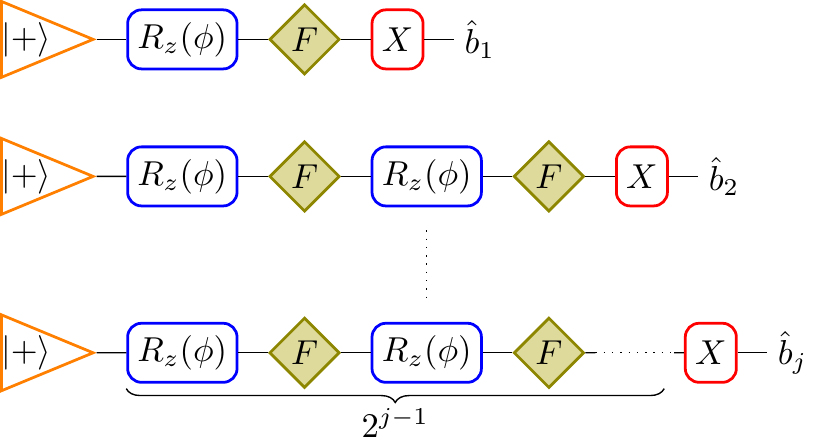}
		\caption{Protocol Ib}\label{fig:protocol2}
	\end{subfigure}
	\quad \quad
	\begin{subfigure}[t]{2.2in}
		\includegraphics[scale=0.7]{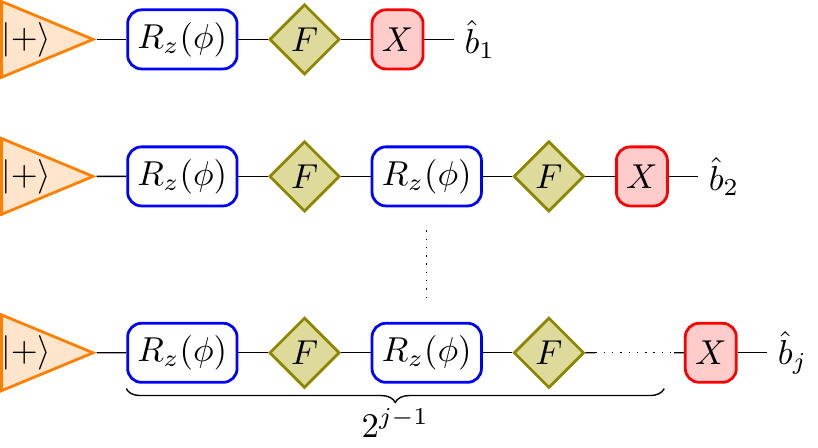}
		\caption{Protocol Ic}\label{fig:protocol3}
	\end{subfigure}
	\caption{Three serial quantum metrology protocols for estimating $j$ bits of the phase $\phi$. Blue boxes denote the field to be sensed, with its allied noise beyond our control. Orange triangles are inputs and red boxes are measurements, both under our control. The protocols start with the state $\ket{+}= (\ket{0}+\ket{1})/\sqrt{2}$ probes. Green rhombuses denote fault tolerance interleaved with sensing the field. Filled shapes denote FT implementations.  
}\label{fig:protocols}
\end{figure*}

\subsection{Encoding}

Quantum Reed-Muller codes (QRMCs) are quantum stabiliser codes constructed from classical Reed-Muller codes RM($r,m$). RM($r,m$) have order $r$ and block length $2^m$ for $0 \leq r \leq m$~\cite{macwilliams1977theory}. 
The QRMC QRM$(1,m)$ has a block size of $2^m-1$ qubits, encodes one qubit and has minimum distance of $3$. 
Transversality of QRMCs enables a logical operation on a logical state by applying transversal gates on the $2^{m}-1$ physical qubits. 
We choose RM$(1,n+1)$ as the basis for the QRMCs, which are transversal for $T_j$, $j \leq n$~\footnote{This choice of $r$ is made on explicit calculations, see Appendix~\ref{appendix_ft}.}. However, these QRMCs are not transversal for $T_j$ for $j>n$. The effect of these post-transversal rotations is subtle and needs to be counteracted in FT metrology.
 Applying $T_n$ transversally on QRM($1,n+1)$ applies the logical $T_n^{\dagger}$ gate. 

\section{Results}
\label{Sec:res}

Quantum phase estimation can be performed in series. It can also be performed in parallel where multiple qubits in an entangled state interrogate the field simultaneously. They perform similarly to serial strategies where a single qubit interrogates the field multiple times coherently.  See Sec.~(\ref{appendix_parallel}).

We introduce fault tolerance into quantum metrology in three stages. 
The first, Protocol Ia (Fig.~(\ref{fig:protocol1})), is affected by noise everywhere but uses no fault tolerance. It serves as our benchmark.
The second, Protocol Ib (Fig.~(\ref{fig:protocol2})), comes in two types - with noiseless and noisy devices but uses fault tolerance to counteract noise in the field only.
The third,  Protocol Ic (Fig.~(\ref{fig:protocol3})), is affected by noise in both the field and devices and uses fault tolerance to counteract them both.
These protocols can be applied to any phase estimation scheme. 
Since different phase estimation schemes perform differently under noise, their FT threshold improvements and resource requirements will be different.

We illustrate the above methodology for a phase estimation scheme due to Rudolph and Grover (RG)~\cite{rudolph2003quantum}, which we choose for its operational simplicity. 
The RG protocol performs bitwise phase estimation, is non-adaptive, and requires only a Pauli $X$ measurement. The original RG protocol cannot estimate all possible phases -- it has an excluded region~\footnote{A different protocol using a mixed radix representation of the phase can avoid the excluded regions~\cite{ji2008parameter}, but its use in our FT methodology (See Appendix \ref{appendix_ji}) is left open for want of a code family that is simultaneously transversal for $\text{diag}(1, e^{i2\pi /2^{n}})$ and $\text{diag}(1, e^{i2\pi /3^{n}}).$} captured by a parameter $\gamma.$ 	We now present our main results.

\textit{No fault tolerance:}  For any bit $b_j,$ we denote its estimate as $\widehat{b}_j.$ 
The RG scheme assumes $0 \leq \phi < \pi,$ whence $\widehat{b}_0=0.$ We use it to estimate the unknown phase $\phi$ to $t$ bits.
This phase estimation protocol labelled Protocol Ia is presented below and depicted in Fig.~(\ref{fig:protocol1}). 
The protocol converges if it outputs the first $t$ bits of $\phi$ with confidence $\epsilon$.

\hrulefill

\textbf{Protocol Ia}

For $j=1, \ldots, t$

\begin{enumerate}

\item Repeat $M$ times:\\
 (i) Prepare $\ket{+}$. \\
(ii) Interrogate field $2^{j-1}$ times.\\
(iii) Measure $X.$

\item Calculate $\widehat{p}_j$ as the fraction of the $+1$ measurement outcomes out of $M$.  If $\widehat{b}_{j-1}=0$ set $\widehat{\phi}_j = \cos^{-1} (2 \widehat{p}_j -1)$ in $[0,\pi]$ , or else in $[\pi,2\pi]$. If

 (i) $\widehat{b}_{j-1} \pi \leq \widehat{\phi}_j <  \widehat{b}_{j-1} \pi + (\pi/2 - \gamma) $, set $\widehat{b}_j=0$.
 
  (ii) $\widehat{b}_{j-1} \pi + (\pi/2 + \gamma) \leq \widehat{\phi}_j \leq \widehat{b}_{j-1} \pi +  \pi $, set $\widehat{b}_j=1$.
  
Otherwise output estimate up to bit $j-1$ and exit.

\item If $j \neq t$ increase $j$ by one and go to step $1$, otherwise exit and output $$\widehat{\phi}=\widehat{b}_1 \frac{\pi}{2} + \widehat{b}_2 \frac{\pi}{4} + \ldots +  \widehat{b}_t \frac{\pi}{2^t} $$
\end{enumerate}

\vskip-0.17in

\hrulefill

In the noiseless case, this protocol converges everywhere except for $\phi$ in between the decision boundaries -- called the excluded region, which depend on $\gamma$ (Fig.~(\ref{fig:estimator1})). In the latter case, we abort the protocol.  The total range of the excluded angles in the worst case, when there are no overlapping excluded regions, is $2 t  \gamma$. 
This and the convergence of Protocol Ia is proven in Sec.~(\ref{AppendixRG}).

In the noisy case, we define the noise threshold as the probability $p$ below which Protocol Ia converges.
We consider full rank noise, which can occur at any point, before, during or after the interaction of the probe with the field, of the form 
\be
\mathcal{E}(\rho)=(1-p)\rho + p(p_x X \rho X + p_y XZ \rho ZX +p_z Z \rho Z),
\label{eq:noise}
\ee
where $1 \geq p, p_{x,y,z} \geq 0$ and add up to one. This incorporates noise parallel ($p_x=p_y=0$), perpendicular ($p_z=p_y=0$) and combinations thereof. 

Several recent works have studied the effect of noise of various ranks on the scaling of precision of phase estimation~\cite{Matsuzaki2011,Chaves2013, Sekatski2017metrologyfulland,Layden17}. 
All our results are valid for all allowed values of $p_{x,y,z}.$

\begin{figure*}
	\begin{subfigure}[t!]{0.31\textwidth}
		\includegraphics[scale=0.16]{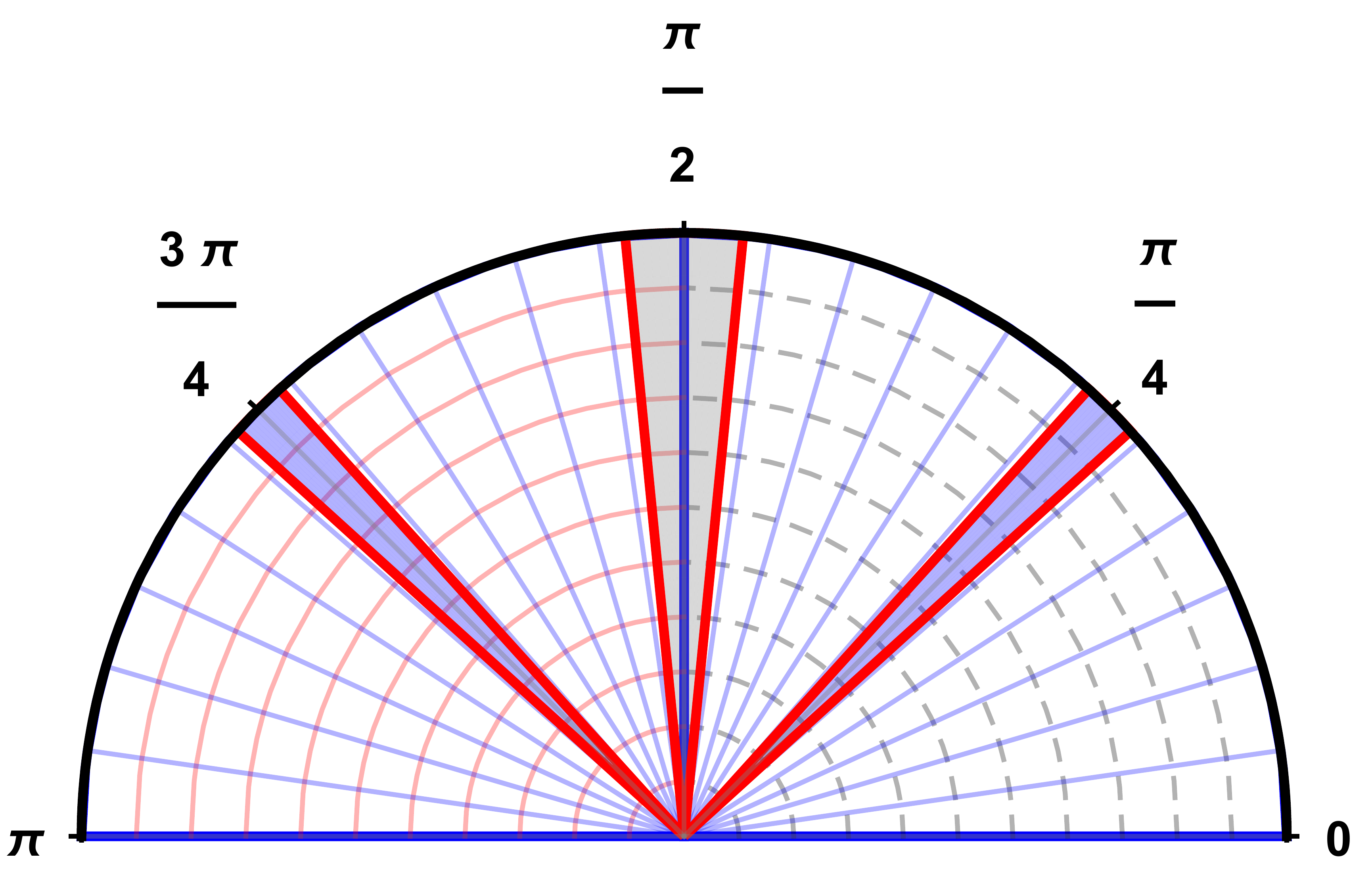}
		\caption{}\label{fig:estimator1}
	\end{subfigure}
		\quad 
	\begin{subfigure}[t!]{0.31\textwidth}
		\includegraphics[scale=0.42]{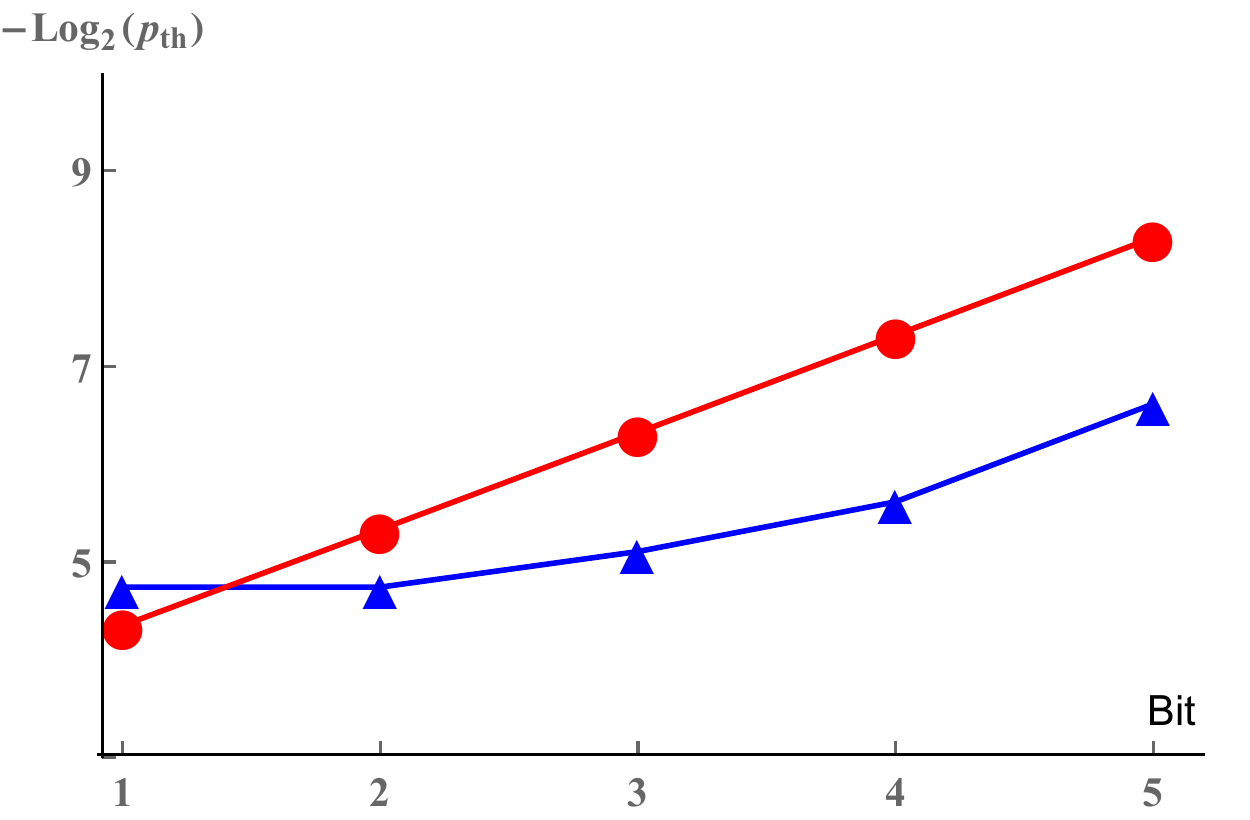}
		\caption{}\label{fig:g32}
	\end{subfigure}
	\quad 
	\begin{subfigure}[t!]{0.31\textwidth}
		\includegraphics[width=\linewidth]{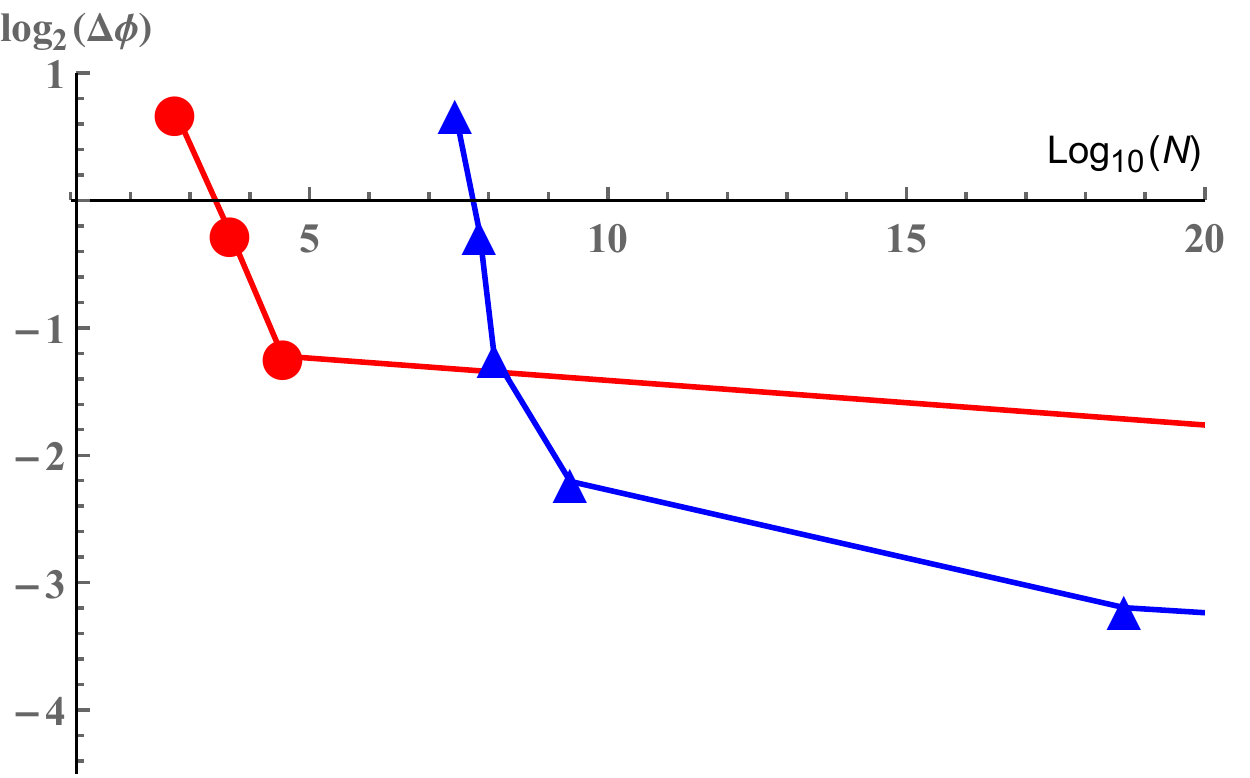}		
		\caption{}\label{fig:resources}
	\end{subfigure}
	\caption{(a) Scheme with $\gamma=\pi/32$. The allowed region for $\phi$ is divided into $[0,\pi/2]$ and $[\pi/2,\pi)$. The decision boundaries are the red lines. For $j=1$ the excluded region (grey) is centred around $\pi/2$. For $j=2$, the excluded region (blue) is centred either around $\pi/4$ if $b_1=0,$ or around $3\pi/4$ if $b_1=1$. This pattern continues for $j>2$.
(b) Noise thresholds for $\gamma=\pi/32$. Red: Protocol Ia; Blue: Protocol Ib.  That the blue line is below the red demonstrates the benefits of FT phase estimation for higher bits. 
(c) Standard deviation \emph{vs} resources  for $\gamma=\pi/32, p=0.63\%$ Red: Protocol Ia. Blue: Protocol Ib, and markers denote bits starting from one and increasing left to right. Improvement from fault tolerance appears in estimating higher bits. See Appendix~\ref{appendix:extra_graphs} for other values of $\gamma.$	}
\label{fig:gammas}
\end{figure*}

Mathematically, Protocol Ia converges for $p<p_{\text{th}},$ where the threshold for the noise affecting the field $p_{\text{th}}$ is the solution of (See Sec.~(\ref{AppendixRG2}))
\be
1-(1-p)^{2^{t-1}} =  \delta(\gamma) = |\sin\gamma|/2.
\label{eqn_phys_thr} 
\ee
The robustness of Protocol Ia against noise depends on $\gamma.$  A larger $\gamma$ excludes more angles but makes the protocol more robust against noise.
Our FT protocols overcome this trade-off.
The threshold $p_{\text{th}}$ obtained from Eqn.~(\ref{eqn_phys_thr}) and presented in Fig.~(\ref{fig:g32}) sets the benchmark against which we compare our next two FT protocols. 
 A larger $p_{\text{th}}$ denotes a greater resilience to noise. 

The number of field interrogations, our resource, required for Protocol Ia to converge is  (See Sec.~(\ref{AppendixRG2}))
\be
N = (2^t-1)  \frac{1}{2\left(\delta(\gamma)-p_f\right)^2} \ln \left(\frac{2t}{\epsilon}\right),
\ee
where $p_f = 1-(1-p)^{2^{t-1}}$. We plot the standard deviation of our estimate $\Delta \phi$ against the resources required for this protocol for a fixed $p$ and $\gamma$ in Fig.~(\ref{fig:resources}). If $p \ll p_{\text{th}}$ for a given $t,$ $\Delta \phi = O(\log N/N),$ where the logarithmic term appears due to bitwise estimation~\cite{rudolph2003quantum} and the $1/N$ term represents the `Heisenberg scaling' in very low noise.

\textit{Error detection against field noise:} First we assume noiseless devices. 
Protocol Ib begins by creating a Bell state between the probe and an ancilla. The probe is then encoded using QRMCs.
The encoded subsystem interrogates the field transversally and is measured in the logical $X$ basis. This, along with appropriate local correction, teleports information of $\phi$ onto the ancillae at the physical level. 
This process, represented by the green rhombuses and blue boxes in Fig.~(\ref{fig:protocol2}) is repeated $2^{j-1}$ times, using the output of one step as the input to the next to estimate $b_{j}$.
Protocol Ib combats noise of the form of Eqn.~(\ref{eq:noise}) in the field using error detection.

\hrulefill

\textbf{Protocol Ib}

For $j=1,\ldots,t$

\begin{enumerate}

\item Repeat $M$ times\\
(i) Prepare probe $\ket{+}$. Set $k=1$.\\
(ii) Prepare ancilla $\ket{0}$. Apply CNOT between probe and ancilla. Encode probe by QRM($1,j+2$). \\
(iii) Interrogate field transversally with probe. Apply error detection on probe. Restart (i) if syndrome measurements reject.\\
(iv) Teleport by measuring probe in logical $X$ and adapting Pauli frame accordingly (See Fig.~(\ref{protocolIbcircuit})).\\
(v) If $k<2^{j-1}, $ increase $k$ by one, use ancilla as new probe and return to (ii). \\ 
(vi) Measure $X$.

\item Step 2 of Protocol Ia with $\gamma$ replaced by $\gamma'.$

\item If $j \neq t$ increase $j$ by one and go to step $1$, otherwise exit and output $$\widehat{\phi}=\widehat{b}_1 \frac{\pi}{2} + \widehat{b}_2 \frac{\pi}{4} + \ldots +  \widehat{b}_t \frac{\pi}{2^t}$$
\end{enumerate}

\vskip-0.15in

\hrulefill

The decision boundaries of Protocol Ib are defined by parameter $\gamma'$- the `logical' version of $\gamma$. This difference arises if $R_z(\phi)$ is not transversal for the QRMCs. 
If $\gamma$ was  the physical rotation, the logical state after step 1(iv) of Protocol Ib undergoes a $Z$-rotation by (Lemma~(\ref{lemma_disc}), Sec.~(\ref{Appendix_disc}))
\be
\gamma'= \gamma - 2 \arctan \left( \frac{\sin(2^{j+1} \gamma)}{(2^{j+2}-1)+\cos(2^{j+1} \gamma)} \right),
\label{eq:gamma}
\ee
For large $j$ this non-transversality has a small effect since $\gamma' = \gamma - O(2^{-j}).$
Following the analysis of Protocol Ia, the range of the excluded angles in the worst case is again $2t \gamma,$ not $2t \gamma'$. 

The probability of logical error in a single interrogation is the probability that the syndrome measurements do not detect the errors and the errors corrupt $X$ measurement. 
Since $\phi$ is unknown we cannot apply a suitable dephasing transformation (also known as twirl) on the noisy states to reduce noise to only $Z$ errors, unlike FT quantum computing~\cite{bravyi2005universal}.  
So we measure both $X$ and $Z$ stabilizers and the corresponding failure probabilities $p^{X}_{\text{err}}$ and $p^{Z}_{\text{err}}$ are given in Eqn.~(\ref{eqn:x_detection}) and~(\ref{eqn:z_detection}), Sec.~(\ref{Appendix_thres}). 
The threshold for $p$ is now obtained by solving $ p_f \equiv 1-(1-p^{X}_{\text{err}})^{2^{j-1}}(1-p^{Z}_{\text{err}})^{2^{j-1}} = \delta(\gamma'),$ which corresponds to  Eqn.~(\ref{eqn_phys_thr}) at the logical level.
This threshold is presented in Fig.~(\ref{fig:g32}).
For higher order QRMCs, $p^{X}_{\text{err}} \ll p^{Z}_{\text{err}}$ below the threshold. 

The number of field interrogations, our resources, required for Protocol Ib to converge depends on $p_{\mathrm{n}}$, the probability of retransmission due to noise and $p_{\mathrm{r}},$ the probability of retransmission due to non-transversality. Using Lemma~(\ref{lemma_costpost}), Sec.~(\ref{Appendix_disc}),
$
p_{\mathrm{r}} = 1 - \left(1-\frac{1}{2^{j+1}}\right)^{(j+2)2^{j-1}}.
\label{eq:smallrot}
$
If the probabilities of passing the $X$ and $Z$ error syndrome measurements for bit $j$ are given by $p_{X\rm{pass}} $ and $p_{Z\rm{pass}} $ respectively (Eqn.~(\ref{eqn:x_pass}) and~(\ref{eqn:z_pass}), Sec.~(\ref{Appendix_thres})),  
$p_{\mathrm{n}} = 1 - \left( p_{X\rm{pass}}  p_{Z\rm{pass}}  \right)^{2^{j-1}}.$
This gives
\be
N= \sum_{j=1}^t 2^{j-1} C(j)\frac{1}{2\left(\delta(\gamma') - p_f\right)^2}  \ln \left(\frac{2t}{\epsilon}\right),
\ee
with $C(j) = (2^{j+2}-1)/(1-p_{\mathrm{n}})(1-p_{\mathrm{r}})$ being the overhead from the QRMC. We plot $\Delta \phi$ versus the resources required -- including extra interrogations due to retransmissions -- for a fixed $p$ and $\gamma$ in Fig.~(\ref{fig:resources}).

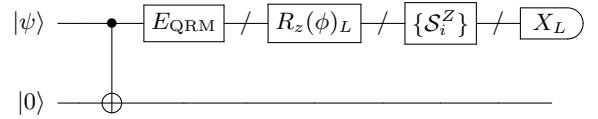
\begin{figure}[t]	
	\mbox{
\Qcircuit  @C=0.9em @R=2em {
		& \lstick{\ket{\psi}}&  \qw	& \ctrl{1}		& \gate{E_{\text{QRM}}}	&  {/} \qw	& \gate{R_z(\phi)_L} &{/}  \qw	&  \gate{\{ \mathcal{S}_i^{Z} \}}	& {/}  \qw	& \measureD{X_L}		& 		& 	 \\
			& \lstick{\ket{0}}& \qw		&\targ	&  \qw		& \qw	&  \qw	&   \qw	&   \qw & \qw	&   \qw	 }
			}
		\caption{Steps 1(i) - 1(iv) of Protocol Ib. For $k=1,$ $\ket{\psi}= \ket{+}$ or else the output of previous $k$. $E_{\text{QRM}}$ is the encoding circuit for QRM($1,j+2$), $R_z(\phi)_L$ is logical (transversal) application of the field, $\{ \mathcal{S}_i^{Z} \}$ are all the $Z$ stabilizer measurements and $X_L$ is logical $X$ measurement from which we extract the $X$ syndromes.}\label{protocolIbcircuit}
\end{figure}

Now we deal with noise in devices, which we assume to be independent of the field noise.
This results in the new threshold equation 
\be
1-(1-p'^{X}_{\text{err}})^{2^{j-1}}(1-p'^{Z}_{\text{err}})^{2^{j-1}}(1-p')^{3 \times 2^{j-1} +2} = \delta(\gamma')
\label{eq:noisydev}
\ee
involving noise of the form of Eqn. (\ref{eq:noise}) for the field ($p$) and the devices ($p'$).
The failure probabilities $p'^{X}_{\text{err}}, p'^{Z}_{\text{err}}$ now have an additional contribution from the noisy devices, which itself includes the effect of noisy non-transversal encoding and noisy syndrome measurements. 
The latter are $E_{\text{QRM}}$ and $ \{\mathcal{S}_i^{Z}\}, X_L$ in Fig.~(\ref{protocolIbcircuit}). 
Since Eqn.~(\ref{eq:noisydev}) involves two variables $p,p',$ there is no unique solution for the two thresholds - $p_{\text{th}}$ and $p'_{\text{th}}.$ 
For small $p'_{\text{th}},$ see Fig. (\ref{fig:threshold_relations}) for improvements in $p_{\text{th}}$ and
Sec.~(\ref{Appendix_noisydevIb}) for details.

\begin{figure}[b]
		\includegraphics[scale=0.65]{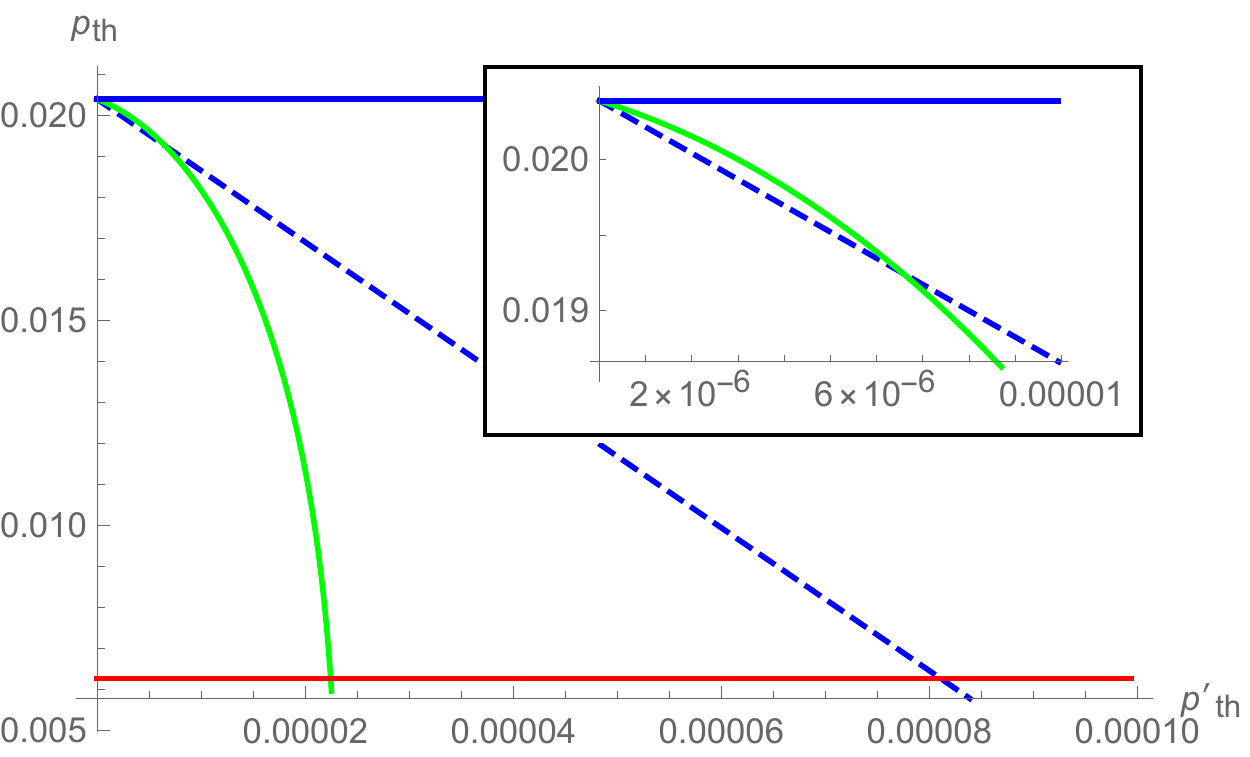}
		\caption{Relationship between thresholds $p_{\text{th}}$ and $p'_{\text{th}}$ for $\gamma=\pi/32$ and $j=4$. Improved threshold for Protocol Ib with device noise (Dashed blue) over Ia (Red). Improved threshold of Protocol Ic (Green) over Protocol Ib with device noise in sub-region enlarged. Protocol Ib without device noise (Solid blue) is provided for reference.
}\label{fig:threshold_relations}
	\end{figure}

\textit{Fault tolerance everywhere:} 
Finally, Protocol Ic (Fig. (\ref{fig:protocol3})) combats noise at any stage of the sensing process.
In quantum computation, the lack of transversal universal gate sets~\cite{eastin2009restrictions} is overcome by either gate distillation or code switching. In metrology, the former is prohibitive because $\phi$ is unknown (See Sec.~(\ref{Appendix_why})).
 Our Protocol Ic proceeds via switching~\cite{1367-2630-17-8-083002} between the QRM($1,3$) Steane code which is transversal for $H$ and higher order QRMCs~\cite{anderson2014fault}, along with the error detection method of Protocol Ib.

\hrulefill

\textbf{Protocol Ic}

For $j=1,\ldots,t$

\begin{enumerate}

\item Repeat $M$ times\\
(i) Prepare $\ket{+_L}$ using FT procedure employing the Steane code and switch to QRM($1,j+2$). Set $k=1$.\\
(ii) Prepare ancilla $\ket{0_L}$ using FT procedure employing QRM($1,j+2$). Apply transversal FT CNOT between probe and ancilla. \\
(iii) Interrogate field transversally with probe. Apply error detection on probe. Restart (i) if syndrome measurements reject.\\
(iv) Teleport by measuring probe in logical $X$ and adapting Pauli frame accordingly  (See Fig.~(\ref{protocolIccircuit})). \\
(v) If $k<2^{j-1}, $ increase $k$ by one, use ancilla as new probe and return to (ii).  \\ 
(vi) FT measurement of logical $X$.

\item Step 2 and 3 of Protocol Ib.
\end{enumerate}

\vskip-0.2in

\hrulefill

Protocol Ic behaves exactly  as Protocol Ib in terms of convergence and resource requirements.
The thresholds for Protocol Ic are given by modifying the failure probabilities $p'^{X}_{\text{err}}, p'^{Z}_{\text{err}}$, the number of points of failure and the noise $p'$ in Eqn.~(\ref{eq:noisydev}). 
Protocol Ic has no non-transversal encoding and failure probabilities just include noisy syndrome measurements (Sec.~(\ref{Appendix_noisydevIc})). 
Compared to Protocol Ib, Protocol Ic now provides a larger improvement in $p_{\text{th}}$ over Protocol Ia, but over a reduced range of $p'$ as shown in Fig.~(\ref{fig:threshold_relations}).
The improvements are limited by the poor QRMC error correction thresholds.
Larger improvements should be attainable if codes with better thresholds and suitable transversality properties can be designed.

\section{Analysis of the parametrised Rudolph-Grover scheme -- convergence, noise resilience and resources}
\label{AppendixRG}

The unknown phase parameter $\phi$ is expressed in a radix $2$ expansion as $\phi = 2 \pi \times 0.b_0 b_1 b_2  \ldots = b_0 \frac{2 \pi}{2} + b_1 \frac{2 \pi}{2^2} + b_2 \frac{2 \pi}{2^3} \ldots$. Setting $b_0=0$ leads to
\be
\phi = b_1 \frac{\pi}{2}+ b_2 \frac{\pi}{4} \ldots .
\ee
We denote our estimate of the unknown parameter after the protocol as $\widehat{\phi}$.

\subsection{Noiseless case}
\label{AppendixRG1}

Assume first that Protocol Ia is implemented in the noiseless case. Let $p_1$ be the probability of obtaining $0$ (eigenvalue $+1$) in our measurements in a noiseless Protocol~Ia for $j=1$. Let $\widehat{p}_1$ be our (real valued) estimate, which comes from averaging over $M$ i.i.d. repetitions. Seeking $|\widehat{p}_1 - p_1| \leq \delta$ leads to
\begin{eqnarray}
\mathrm{prob}(|\widehat{p}_1 - p_1 | \leq \delta ) \geq 1 - 2 e^{-2M \delta^2}    \nonumber 
\end{eqnarray}
from the Hoeffding inequality. 

Let us choose $\delta =  |\cos^2 (\frac{\pi}{4}) - \cos^2 (\frac{\pi}{4} - \frac{\gamma}{2})  | = |\sin(\gamma)/2|$. For $\gamma = \pi/8$, $\delta \approx 0.191,$ for $\gamma = \pi/32$, $\delta \approx 0.049$.
This implies that, for angle $\phi$ in the allowed region $[0,\pi/2-\gamma]$ or $[\pi/2+\gamma,\pi]$, if\\
(i) $0 \leq \widehat{\phi}_1 <   (\pi/2 - \gamma),$
\be
\mathrm{prob}\left(0 \leq \phi  < \frac{\pi}{2}\right) \geq 1 -    2 e^{-2M \delta ^2}
\ee
and $\mathrm{prob}(\widehat{b}_1 = b_1 =0 )$ is equally high;\\
(ii) $(\pi/2 + \gamma) < \widehat{\phi}_1 \leq  \pi ,$
\be
\mathrm{prob}\left(\frac{\pi}{2} \leq \phi  < \pi\right) \geq 1 -    2 e^{-2M \delta  ^2}  \nonumber
\ee
and $\mathrm{prob}(\widehat{b}_1 = b_1 =1 )$ is equally high. This concludes the analysis for $j=1$.

Assuming that the estimation for $j=1$ was correct, we proceed to the estimation of the other bits. We use induction to calculate all the conditional probabilities. Suppose all bits $b_k$, $k<j$ are correctly estimated. 
The probe after the $2^{j-1}$ consecutive interrogations is $(\ket{0}+e^{i \phi_j} \ket{1})/\sqrt{2}$, where $\phi_j = 2^{j-1} \phi = 2 \pi \times 0.b_{j-1} b_j b_{j+1} \ldots = b_{j-1} \frac{2 \pi}{2} + b_j \frac{2 \pi}{2^2} + b_{j+1} \frac{2 \pi}{2^3} \ldots$, where $b_{j-1}$ is known from previous estimation.

Again, using the Hoeffding inequality, we bound the probability of having error smaller than the same parameter $\delta$. The allowed region for $\phi_j - b_{j-1} \pi$ should be again 
$[0,\pi/2-\gamma]$ or $[\pi/2+\gamma,\pi]$, and if \\
(i) $\widehat{b}_{j-1} \pi \leq \widehat{\phi}_j <  \widehat{b}_{j-1} \pi + (\pi/2 - \gamma),$
\be
\mathrm{prob}\left(0 \leq \phi_j - b_{j-1} \pi  < \frac{\pi}{2}\right) \geq 1 -    2 e^{-2M \delta ^2}, \nonumber
\ee
and $\mathrm{prob}(\widehat{b}_j = b_j =0 )$ is equally high, conditioned on the estimations of prior bits being correct; \\
(ii) $\widehat{b}_{j-1} \pi + (\pi/2 + \gamma) \leq \widehat{\phi}_j \leq \widehat{b}_{j-1} \pi +  \pi $,
\be
\mathrm{prob}\left(\frac{\pi}{2} \leq \phi_j - b_{j-1} \pi  < \pi\right) \geq 1 -    2 e^{-2M \delta  ^2},  \nonumber
\ee
and $\mathrm{prob}(\widehat{b}_j = b_j =1 )$ is equally high, conditioned on the estimations of prior bits being correct. This concludes our analysis for $j$.

The probability that all the bits up to $b_t$ being estimated correctly is lower bounded by $(1-2 e^{-2M \delta  ^2})^t \geq 1-  2 t e^{-2M \delta  ^2}$. To have a maximum error $\epsilon$ in our estimator to be correct up to the $t$-th bit, we choose $M$ such that $\epsilon \geq 2 t e^{-2M \delta  ^2}$. This leads to
\be
M \geq \frac{1}{2\delta^2} \ln \left(\frac{2t}{\epsilon}\right).
\ee
The total overhead in uses of the field to estimate $\phi$ to $t$ bits with error $\epsilon$ becomes
\be
N = \frac{2^{t}-1}{2\delta^2} \ln \left(\frac{2t}{\epsilon}\right).
\ee

The allowed angles for which the above convergence arguments hold are as follows. 
From the analysis above, the estimation of the first bit converges with high probability if $\phi \in [0,\pi/2-\gamma] \cup [\pi/2+\gamma,\pi]$. Thus the length of the excluded region is $2\gamma$.
For the second bit, consider estimating $\phi_2=2\phi.$ If $b_1=0$, $\phi \in [0,\pi/4-\gamma/2] \cup [\pi/4+\gamma/2,\pi/2]$ and if $b_1=1$, $\phi \in [\pi/2,\pi/2 + \pi/4-\gamma/2] \cup [\pi/2 + \pi/4+\gamma/2,\pi]$. Length of the excluded region is again $2\gamma$.

In general, consider estimating $\phi_j=2^{j-1} \phi \mod 2 \pi.$ Suppose $b_1= \ldots =b_{j-2}=0.$ If $b_{j-1}=0$,
\be
\phi \in \left[0,\frac{\pi}{2^j}-\frac{\gamma}{2^{j-1}}\right] \cup \left[\frac{\pi}{2^j}+\frac{\gamma}{2^{j-1}},\frac{\pi}{2^{j-1}}\right], \nonumber
\ee and if $b_{j-1}=1,$
\be
\phi \in  \left[ \frac{\pi}{2^{j-1}}, \frac{\pi}{2^{j-1}} + \frac{\pi}{2^j} - \frac{\gamma}{2^{j-1}} \right] \cup \left[ \frac{\pi}{2^{j-1}} + \frac{\pi}{2^j}+ \frac{\gamma}{2^{j-1}}, \frac{\pi}{2^{j-2}} \right]. \nonumber
\ee
Continuing with the $2^{j-2}$ possibilities for $b_1, \ldots, b_{j-2}$, each of which exclude regions of length $\gamma/2^{j-2},$ we obtain a total excluded region of length $2\gamma$.
In the worst case, of the regions not being overlapping, the excluded region has total angle $2 t  \gamma$.

\subsection{Noisy case}
\label{AppendixRG2}

We now consider the noisy case and denote the probability of an error occurring in an interrogation step as $p$. Then, the probability $p_f(p,j)$ of the measurement result being incorrect after a number of interrogations and a final measurement depends on $p$ and the number of interrogations (which depends on $j$). In Protocol Ia, we undertake $2^{j-1}$ interrogations, whereby $p_f$ is upper bounded by $1-(1-p)^{2^{j-1}}$. 

The following analysis holds for any $j$.
Let $p_j$ be the probability of obtaining $0$ (eigenvalue $+1$) if there was no noise. With probability $p_f,$ this result we get will be flipped. Let $p'_j$ be the 'noisy' probability of obtaining $0$. Then
\be
p'_j = p_j (1-p_f) + (1-p_j) p_f,
\ee
whereby $p'_j - p_j = p_f(1-2p_j) $, implying 
\be
|p'_j - p_j | \leq p_f.
\ee
After $M$ repetitions, the Hoeffding inequality gives the noisy estimate $\widehat{p'}_j$ as
\be
\mathrm{prob} \left(|\widehat{p}'_j - p'_j | \geq \delta \right) \leq 2 e^{-2M \delta^2} .
\ee
Adding $|p'_j - p_j |$ gives
\be
\mathrm{prob} \left(|\widehat{p}'_j - p'_j | + |p'_j - p_j | \geq \delta + |p'_j - p_j | \right)\leq 2 e^{-2M \delta^2} . \nonumber
\ee
We then use the fact that $(\text{prob}(x \geq b) \leq c)\wedge(y\leq x) \Rightarrow \text{prob} (y \geq b) \leq c$, which can be proven by writing the probabilities as integrals and changing variables. Since $ |\widehat{p}'_j - p_j | \leq |\widehat{p}'_j - p'_j | + |p'_j - p_j |$,
\be
\mathrm{prob}(|\widehat{p}'_j - p_j | \geq \delta + |p'_j - p_j | )\leq 2 e^{-2M \delta^2} .  \nonumber
\ee
Thus,
\be
\mathrm{prob}(|\widehat{p}'_j - p_j | \geq \delta)\leq 2 e^{-2M (\delta - |p'_j - p_j | ) ^2} \leq    2 e^{-2M (\delta - p_f ) ^2},  \nonumber
\ee
whereby
\be
\mathrm{prob}\left(|\widehat{p}'_j - p_j | < \delta \right) > 1 -    2 e^{-2M (\delta - p_f ) ^2}.  \nonumber
\ee
We therefore get confidence in our estimation for the $j$-th bit only if $p_f < \delta,$ in which case the same proof of convergence holds as in the noiseless case. This means that there is a probability $p$ of failure in a single interrogation above which the protocol does not converge and is given by the solution of $1-(1-p)^{2^{t-1}} = \delta = |\cos^2 (\frac{\pi}{4}) - \cos^2 (\frac{\pi}{4} - \frac{\gamma}{2})  | $ for a fixed $\gamma$ and $t$. We call this the noise threshold $p_{\text{th}}$ of the protocol.

Following the same analysis as before and replacing $\delta$ by $\delta-p_f(p,t)$ we have
\be
N = \frac{2^{t}-1}{2(\delta-p_f)^2} \ln \left( \frac{2t}{\epsilon} \right). \nonumber
\ee

\paragraph*{Standard deviation:}
A canonical way of quantifying the performance of an estimator is its standard deviation $\Delta \phi$. We derive this for a fixed $\epsilon$ adapting the technique from Ref.~\cite{de2005quantum}. At the conclusion of the estimation protocol, with probability $1-\epsilon$ an estimate $\phi_{est}$ is obtained which is the correct one up to $t$ bits of precision ($\Delta \phi_{est} \leq \pi/2^{t+1}$). Otherwise, we get a random estimate $\phi_{r},$  which we assume to be independent of $\phi_{est}$ to ease our analysis. Thus $\phi = (1-\epsilon) \phi_{est} + \epsilon \phi_{r},$ and
\ben
\Delta \phi &=& \sqrt{(1-\epsilon)^2 (\Delta\phi_{est})^2  + \epsilon^2 (\Delta\phi_{r})^2} \nonumber \\
		&=& \sqrt{(1-\epsilon)^2 \frac{\pi^2}{2^{2(t+1)}}  + \epsilon^2 \pi^2} . \nonumber
\een
We choose $\epsilon$ so that $\Delta \phi$ decreases inversely with the largest possible function of the total overhead. Let $\epsilon = 1/2^{t}$. Since $\Delta \phi = O(2^{-t})$ for noise significantly smaller than the threshold, $N = O(t 2^t)$ , and we have $\Delta \phi = O(\log N/N),$ ignoring terms logarithmic in $t$.

\section{Non-transversality effects in QRMCs}
\label{Appendix_disc}

We provide results for QRMCs for the effect of applying transversally gates that are non-transversal for a particular QRMC, under postselection for the correct syndrome outcomes.
The equations for bit $j$ in our protocols are obtained by setting $m=j+2$ in the following Lemmas.

\begin{lemma}
Apply $R_z(-\phi)$ transversally, where $\phi=0.b_0 b_1 b_2 \ldots$, on a logical single qubit state encoded by code QRM($1,m$). Postselecting on accepting the syndromes creates, up to a global phase, a logical rotation of\begin{eqnarray}
\phi'= \phi - 2 \arctan \left( \frac{\sin(\phi_{m})}{(2^m-1)+\cos(\phi_m)} \right),
\label{eq:diff}
\end{eqnarray}
around the $Z$ axis, where $\phi_m = 2^{m-1} \phi = 2 \pi \times 0.b_{m-1} b_{m} b_{m+1} \ldots.$
\label{lemma_disc}
\end{lemma}

\begin{proof}
Let
\be
P_{+1} = \prod_{i=1}^{2^m-m-2} \frac{ (I+S^Z_{i})}{2^{2^m-m-2}} \prod_{i=1}^{m} \frac{ (I+S^X_{i})}{2^m}
\ee
be the projector to the code space, i.e. the positive eigenspace of the $Z$ and the $X$ syndrome measurements $S^Z_i$ and $S^X_i$ respectively.
The effect of applying $R_z(-\phi)$ transversally and projecting to $P_{+1}$ on logical state $\ket{0_L}$ leads to $P_{+1} R_z(-\phi)^{\otimes 2^m-1} \ket{0_L}$ which is
\be
\frac{1}{\sqrt{2^m}} P_{+1} \left (\ket{\boldsymbol{0}} +e^{-i 2^{m-1} \phi} \!\!\sum_{x \in \overline{RM}\setminus \{ \boldsymbol{0} \} } \ket{x} \right). \nonumber
\ee
The projections coming from the $Z$ stabilizer measurements have no effect on the state. The elements $S^X_i$ correspond to the generators of  the $\overline{RM}$ code (by replacing the $1$'s with $X$'s and the $0$'s with $I$'s) and therefore $\prod_{i=1}^{m}  (I+S^X_{i})$ gives a sum over stabilizers that correspond to all $x \in \overline{RM}$. When applied to the above state they map each codeword to the sum of all  codewords in the code and therefore create the same (global) phase:
\be
1+ (2^m-1) e^{-i \phi_m}  = 1 + (2^m-1)\cos\phi_m + i (2^m-1)\sin\phi_m \nonumber
\ee
where $\phi_m = 2^{m-1} \phi.$ Similarly for the logical $\ket{1_L}$ state, we get $P_{+1} R_z(-\phi)^{\otimes 2^m-1} \ket{1_L}$
\be
\frac{1}{\sqrt{2^m}} P_{+1} \left(e^{-i (2^{m}-1) \phi}\ket{\boldsymbol{1}} + e^{-i (2^{m-1}-1) \phi} \sum_{x \in \overline{RM}} \ket{x + \boldsymbol{1}} \right) \nonumber
\ee
Again, the projectors from the $X$ measurements mix all the phases leading to a global phase of 
\be
e^{-i (2^{m}-1) \phi} + (2^m-1) e^{-i (\phi_m- \phi)}. \nonumber
\ee
Therefore the whole operation adds between the computational basis states a relative phase of 
\begin{eqnarray}
 \phi' &=& \arctan \left( \frac{(2^m -1) \sin \phi_m}{1+ (2^m-1) \cos\phi_m}\right) \nonumber \\
 &-& \arctan \left( \frac{\sin ((2^m -1) \phi) + (2^m-1) \sin (\phi_m- \phi) }{\cos ((2^m -1) \phi) + (2^m-1) \cos (\phi_m- \phi)} \right), \nonumber 
\end{eqnarray}  
wherefrom Eqn.~(\ref{eq:diff}) emerges via trigonometry. 
\end{proof}

The cost of postselection for rotations that are not transversal for QRM$(1,m)$ is given below.

\begin{lemma}
\label{lemma_costpost}
The probability of failure in any of the  $S^X_i$ syndrome measurements on a QRM($1,m$)-encoded state, on which transversal $R_z(-\phi)$ has been applied, is less or equal to $1 -\left(1-\frac{1}{2^{m-1}}\right)^{m}$.
\end{lemma}

\begin{proof}
The probability of failure in the post-selection of each of the $m$  syndromes is at most $\frac{1}{2^{m-1}}$, for any real rotation. This comes from calculating the probability of getting result $0$ in measurement $i.$ This is given by
\begin{eqnarray}
p_{i} = \bra{\chi_{i-1}}\frac{I+S_i^X}{2} \ket{\chi_{i-1}} =  \frac{1}{2} + \frac{\bra{\chi_{i-1}}S_i^X\ket{\chi_{i-1}}}{2} \nonumber \\
\geq  \frac{1}{2} + \frac{2^{m+1}-8}{2 \times 2^{m+1}} = 1 - \frac{1}{2^{m-1}} \nonumber
\end{eqnarray}
for $\ket{\chi_{i-1}}$ being the state that comes after syndrome measurement $i-1$ (renormalized) and  $\ket{\chi_{0}} = R_z(-\phi)^{\otimes 2^m-1} (\ket{0_L} + e^{i \psi} \ket{1_L})/\sqrt{2}$, for some $\psi$. The key observation here is that $S_i^X \ket{\chi_{i-1}}$ is a permutation of the sums of kets of $\ket{\chi_{i-1}}$, where each sum of kets comes from applying $\prod_{j=1}^{i-1}  (I+S^X_{j})$ on each ket of the initial state $\ket{\chi_{0}}$ when written in the physical representation.
 
The probability of failure of all the $m$ $X$ syndromes -- $X$ syndromes are the only ones that potentially reject -- is therefore: $ 1 -(1-\frac{1}{2^{m-1
}})^{m}$. This creates an extra overhead in the resource count.
\end{proof}

\section{Protocol Ib with noiseless devices  -- error detection failure probabilities}
\label{Appendix_thres}

In order to caclulate the thresholds and resources of Protocol Ib we need to find the probability that the error detection procedure fails at each step $k$. We exploit the idea of only error detecting for the errors, followed by the decoding of the code subspace to a Hilbert space of one qubit, in order to get improved thresholds~\cite{kitaev1995quantum}. An instance of the circuit used for error detection at each step $k$ of Protocol Ib (Fig.~(\ref{protocolIbcircuit})) is given for $m=4$ in Fig.~(\ref{teleport_FT}).

For Protocol Ib, unlike in magic state distillation in quantum computing (for more details see~\cite{fujii2015quantum}), the circuit of Fig.~(\ref{teleport_FT}) is applied on a physical level. 
 In Protocol Ib errors only enter through the $R_z (\phi)$ gates and of the form of Eqn.~(\ref{eq:noise}).
Rejections after the syndrome measurements can happen either because of noise or because of the non-transversal effects analysed in Sec.~(\ref{Appendix_disc}). There is no dependency between the two sources of rejection and thus we restrict our analysis here to rejections due to noise.

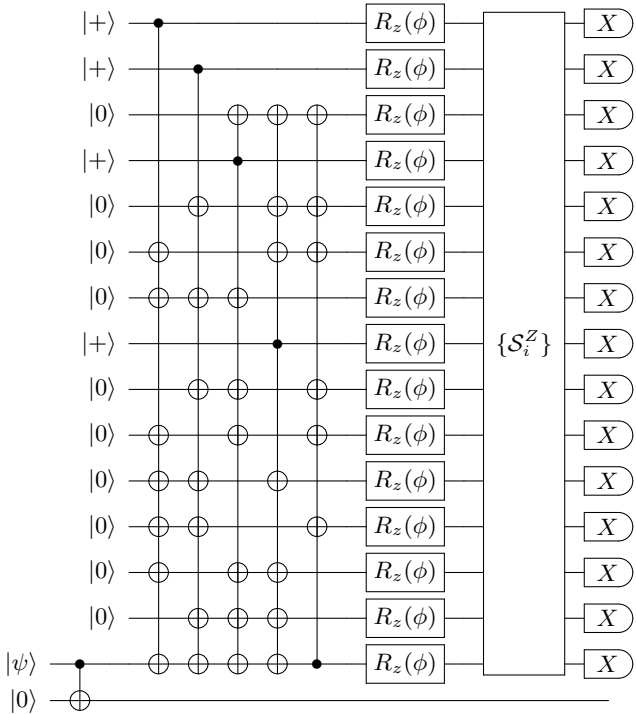
\begin{figure}[t]
\mbox{
\Qcircuit  @C=0.8em @R=0.3em {
&	&	& \lstick{\ket{+}} &\ctrl{14}	& \qw & \qw & \qw & \qw & \qw & \gate{R_z (\phi)}& \qw & \multigate{14}{\{ \mathcal{S}^Z_i \}} & \measureD{X} \\
&	&	& \lstick{\ket{+}} & \qw		& \ctrl{13} & \qw & \qw & \qw & \qw & \gate{R_z (\phi)}& \qw &  \ghost{\{ \mathcal{S}^Z_i \}} & \measureD{X} \\
&	&	& \lstick{\ket{0}} & \qw		& \qw & \targ & \targ & \targ & \qw & \gate{R_z (\phi)}& \qw &  \ghost{\{ \mathcal{S}^Z_i \}} & \measureD{X} \\
&	&	& \lstick{\ket{+}} & \qw		& \qw & \ctrl{-1}  \qwx[11] & \qw & \qw & \qw &\gate{R_z (\phi)}& \qw  &  \ghost{\{ \mathcal{S}^Z_i \}} & \measureD{X} \\
&	&	& \lstick{\ket{0}} & \qw		& \targ & \qw & \targ & \targ & \qw &\gate{R_z (\phi)}& \qw &  \ghost{\{ \mathcal{S}^Z_i \}} & \measureD{X} \\
&	&	& \lstick{\ket{0}} & \targ		& \qw & \qw & \targ & \targ & \qw &\gate{R_z (\phi)}& \qw &  \ghost{\{ \mathcal{S}^Z_i \}} & \measureD{X} \\
&	&	& \lstick{\ket{0}} & \targ		& \targ & \targ & \qw & \qw & \qw &\gate{R_z (\phi)}& \qw &  \ghost{\{ \mathcal{S}^Z_i \}} & \measureD{X} \\
&	&	& \lstick{\ket{+}} & \qw		& \qw & \qw & \ctrl{-5} \qwx[7] & \qw & \qw &\gate{R_z (\phi)}& \qw &  \ghost{\{ \mathcal{S}^Z_i \}} & \measureD{X} \\
&	&	& \lstick{\ket{0}} & \qw		& \targ & \targ & \qw & \targ & \qw &\gate{R_z (\phi)}& \qw &  \ghost{\{ \mathcal{S}^Z_i \}} & \measureD{X} \\
&	&	& \lstick{\ket{0}} & \targ		& \qw & \targ & \qw & \targ & \qw &\gate{R_z (\phi)}& \qw &  \ghost{\{ \mathcal{S}^Z_i \}}& \measureD{X} \\
&	&	& \lstick{\ket{0}} & \targ		& \targ & \qw & \targ & \qw & \qw &\gate{R_z (\phi)}& \qw &  \ghost{\{ \mathcal{S}^Z_i \}}& \measureD{X} \\
&	&	& \lstick{\ket{0}} & \targ		& \targ & \qw & \qw & \targ & \qw &\gate{R_z (\phi)}& \qw &  \ghost{\{ \mathcal{S}^Z_i \}}& \measureD{X} \\
&	&	& \lstick{\ket{0}} & \targ		& \qw & \targ & \targ & \qw &\qw & \gate{R_z (\phi)}& \qw &  \ghost{\{ \mathcal{S}^Z_i \}} & \measureD{X} \\
&	&	& \lstick{\ket{0}} & \qw		& \targ & \targ & \targ & \qw & \qw &\gate{R_z (\phi)}& \qw &  \ghost{\{ \mathcal{S}^Z_i \}} & \measureD{X} \\
 \lstick{\ket{\psi}} &\ctrl{1}	& \qw & \qw & \targ		& \targ & \targ & \targ & \ctrl{-12} & \qw &\gate{R_z (\phi)}& \qw &  \ghost{\{ \mathcal{S}^Z_i \}} & \measureD{X} \\
 \lstick{\ket{0}} & \targ	& \qw & \qw &\qw  & \qw & \qw		& \qw & \qw & \qw & \qw & \qw  & \qw  & \qw 
}
}
\caption{FT application of  transversal $R_z (\phi)$ using QRM($1,4$) and teleportation onto input state $\ket{\psi}$.}
\label{teleport_FT}
\end{figure}

Failure comes when the logical outcome of the $X$ measurement is flipped in the case  of no syndrome error is being detected. The failure probabilities at the syndrome detection for $X$ or $Z$ errors, $p_{\text{err}}^X$ and $p_{\text{err}}^Z$ respectively, are

\be
p_{\text{err}}^{X} = p(\text{error}|\text{X pass}) = \frac{p(\text{error},\text{X pass})}{p(\text{X pass})}  \nonumber 
\ee

and similarly,

\be
p_{\text{err}}^{Z} = p(\text{error}|\text{Z pass}) = \frac{p(\text{error},\text{Z pass})}{p(\text{Z pass})} . \nonumber 
\ee

First we focus on the stabilizers that detect the Pauli $X$ errors. These correspond to the rows of the parity check matrix $H_{z}$ of the RM$^*$ code. The undetected noise operators correspond to the codewords of the RM$^*$ code, including noiseless case which corresponds to $\boldsymbol{0}$, given by $V_{H_z}^{\perp}$. Thus 

\begin{equation}
\label{eqn:x_pass}
p(X \text{ pass})=W_{V_{H_z}^{\perp}}(1-p,p),
\end{equation}

where $W_{V}(x,y) = \sum_{c \in V} x^{n-\text{wt}(c)} y^{\text{wt}(c)}$ is the weight polynomial of $V \in \text{GF}(2^n)$ and $\text{wt}(c)$ is the number of ones in the codeword $c$. We can the write the probability of retransmission due to Pauli $X$ noise as 
$ p^X_{\mathrm n} = 1 - W_{V_{H_z}^{\perp}}(1-p,p).$

The above undetected operators could potentially corrupt the logical $X$ measurement if they happen either before  or during the  application of $R_z(\phi)$ signal. To understand this we represent the signal plus noise operation as $R_z(\theta) X R_z(\phi -\theta)$ for some angle $\theta \leq \phi$. Up to global phase this is equal to $R_z(2 \theta) X R_z(\phi)$, thus equal to the original signal plus a Pauli $X$ that has no effect on the logical $X$ measurement, plus an extra term $R_z(2 \theta)$ that can corrupt the logical $X$ measurement when $\theta \neq 0$. From discretization of errors (\cite{nielsen2010quantum}, Theorem 10.2) and the fact that QRMCs can recover from $Z$ noise, these non-Pauli errors $R_z(2 \theta)$ are detected by the $X$ stabilizer measurements unless they correspond to codewords of the Hamming code, and the latter corrupt the logical measurement only when they have odd parity. Since the weights of the codewords that are excluded by these refinements are large, their contribution in the error probability is negligible and therefore we can include in our calculation all codewords of the RM$^*$ code except identity. Therefore, 
\be
p^X_{\text{err}}= \frac{W_{V_{H_z}^{\perp}}(1-p,p) - (1-p)^{2^m-1}}{W_{V_{H_z}^{\perp}}(1-p,p)} 
\ee

Using the codeword weights of $RM^*$ from  Appendix~\ref{appendix_ft}, we obtain 
\begin{widetext}
\be
\label{eqn:x_detection}
p^X_{\text{err}} = \frac{(2^m-1)(1-p)^{2^{m-1}}p^{2^{m-1}-1}+(2^m-1)(1-p)^{2^{m-1}-1}p^{2^{m-1}} +p^{2^m-1}}{(1-p)^{2^m-1}+(2^m-1)(1-p)^{2^{m-1}}p^{2^{m-1}-1}+(2^m-1)(1-p)^{2^{m-1}-1}p^{2^{m-1}} +p^{2^m-1}}.
\ee
\end{widetext}
The results for bit $j$ are obtained by setting $m=j+2$.
Given the form of noise of Eqn.~(\ref{eq:noise}) and that $X$ error detection is made first, the single qubit $X$ error probability is $p(p_x +p_y)$. However, since the function in Eqn.~(\ref{eqn:x_detection}) is  monotonically increasing in $p$, we can replace $p(p_x +p_y)$ by $p$ and get an upper bound $\forall ~p_x,$ $p_y$.

The stabilizers that detect the Pauli $Z$ errors correspond to the rows of the parity check matrix $H_x$ of the dual of the $\overline{RM}$ code, which is the Hamming code $(2^m-1,2^m-1-m,3)$. The undetected noise operators correspond to the codewords of the Hamming code, including noiseless case which corresponds to $\mathbf{0}$, given by $V_{H_x}^{\perp}$. Thus

\begin{equation}
\label{eqn:z_pass}
p(Z \text{ pass})=W_{V_{H_x}^{\perp}}(1-p,p)
\end{equation}

 and the probability of retransmission due to Pauli $Z$ noise is $p^Z_{\mathrm n} = 1 - W_{V_{H_x}^{\perp}}(1-p,p).$

The subset of undetected operators that lead to an error in the logical $X$ measurement are those which anti-commute with the tensor product of $X$ operators: the ones with odd parity.
From duality, the parity matrix $H_{\overline{RM}}$ of the $\overline{RM}$ code is the generator of the codewords of the Hamming code. 
The subset of odd codewords is obtained by complementing the code  generated by the parity check matrix $H_z$ of RM$^*$, which is the same as $H_{\overline{RM}}$ without the $\mathbf{1}$ row, thus keeping only its even generators. Thus
\be
p^Z_{\text{err}}= \frac{W_{V_{H_z}}(p,1-p)}{W_{V_{H_x}^{\perp}}(1-p,p)},
\ee

Using the MacWilliams identity $W_V(x,y)=\frac{1}{|V|}W_{V^{\perp}} (x+y,x-y),$ we obtain
\be
p^Z_{\text{err}}= \frac{|V_{H_x}| W_{V_{H_z}^{\perp}}(1,2p-1)}{|V_{H_z}|W_{V_{H_x}}(1,1-2p)}.
\ee
Using the codeword weights of $RM^*$ and $\overline{RM}$ from  Appendix~\ref{appendix_ft} and $|V_{H_x}|/|V_{H_z}|=2^m/2^{m+1}=1/2$, we obtain 

\begin{widetext}
\be
\label{eqn:z_detection}
p^Z_{\text{err}} = \frac{1+(2^m-1)(2p-1)^{2^{m-1}-1}+(2^m-1)(2p-1)^{2^{m-1}}+(2p-1)^{2^m-1}}{2(1+(2^m-1)(1-2p)^{2^{m-1}})}.
\ee
\end{widetext}

Again, the above is an upper bound on the failure probability due to $Z$ errors, when noise is of the form of Eqn.~(\ref{eq:noise}), for all values of $p_z$.

\section{Noisy devices}
\label{Appendix_noisydev}

\subsection{Protocol Ib: Noisy device thresholds}
\label{Appendix_noisydevIb}

If we assume noisy devices in Protocol Ib, by allowing any device to have noise of the form of Eqn.~(\ref{eq:noise}) with probability $p$ replaced by the device noise probability $p'$, the threshold calculation is different. Failure probabilities of the detection procedure for $X$ and $Z$ errors are denoted by $p'^{X}_{\text{err}}$ and   $p'^{Z}_{\text{err}}$ respectively. These probabilities are given by replacing  probability $p$ by $p+\text{devIb}(p')$ in Eqns.~(\ref{eqn:x_detection}) and~(\ref{eqn:z_detection}) respectively. Probability
$\text{devIb}(p')$ captures the effect of all device noise (except for state preparation and CNOT error which are included separately in the last term on the LHS of Eqn. (\ref{eq:noisydev2})) on one qubit in the detection procedure and is given by
\begin{equation}
\text{devIb}(p') = (c_{\text{e}} + (2^{m}-m-2)+1)p'
\end{equation}
where $c_{\text{e}}$ is the number of points of failure in the non-transversal encoding procedure $E_{\text{QRM}}$ that affect one qubit. The operations in the encoding procedure correspond to the generator matrix of $RM^*(1,m)$. On average, there are approximately $(m +1) 2^{m - 1} / (2^m - 1)$ points per qubit where the entangling operations apply on the particular qubit.
Since each entangling operation in the coding involves approximately $2^{m - 1}$ qubits, we have
\begin{equation}
\langle c_{\text{e}} \rangle \approx \frac{(m + 1) 2^{m - 1}}{2^m -1} 2^{m - 1}.
\end{equation}
For our protocol, we need to set $m=j+2$ in the previous two equations.

The failure probability at the output of Protocol Ib with device noise is bounded away from one by the joint probability that in all of the $2^{j-1}$ rounds, both the state preparation and CNOT are correct and the detection procedure does not fail. 
The latter joint probability can be written as the product of the probability of correct state preparation/CNOT and the probability of detection not failing conditioned on correct state preparation/CNOT. The points of failure for state preparation/CNOT are $3 \times 2^{j-1}+2$. This includes initial probe preparation and Hadamard, as well as ancilla preparation and CNOT ($2$ points of failure) at each interrogation step. Notice that performing the teleportation correction can be avoided by updating the Pauli frame.  Then the threshold equation becomes
\be
1-(1-p'^{X}_{\text{err}})^{2^{j-1}}(1-p'^{Z}_{\text{err}})^{2^{j-1}}(1-p')^{3 \times 2^{j-1} +2} = \delta(\gamma')
\label{eq:noisydev2}
\ee
Since Eqn.~(\ref{eq:noisydev2}) involves two variables $p,p',$ there is no unique solution but rather a relation for the two thresholds - $p_{\text{th}}$ and $p'_{\text{th}}$ - as depicted in Fig.~(\ref{fig:threshold_relations}).

\subsection{Protocol Ic: Why code switching?}
\label{Appendix_why}

In Protocol Ic we combat noise that can enter at any stage of the phase estimation protocol, in interrogating the field, as well as probe and ancilla preparation, entangling gates and measurements.

As in quantum computing, we need to employ some extra encoding throughout the protocol. If we use transversal quantum codes, the same encoding cannot be used everywhere since there is no quantum code transversal for a universal set of gates~\cite{eastin2009restrictions}. Two techniques are known to solve this issue: gate (or state) distillation and code switching. First we explain why the first technique is prohibitively expensive in terms of our resources for phase estimation.

\paragraph*{Gate distillation:} Everything is performed on an underlying quantum error correcting code which is transversal only for Clifford operations (e.g. QRM($1,3$), also known as the Steane  code). The non-Clifford operations are performed by injecting into this code special states, sometimes called magic states, and  then applying a distillation procedure using a higher order QRMC to reduce their noise~\cite{bravyi2005universal}.

In our case, the non-Clifford part of the computation is the $R_z(\phi)$ rotation. In metrology however $\phi$ is unknown. Similarly to Ref.~\cite{bravyi2005universal}, we could inject a state on which the $R_z(\phi)$ rotation has been applied and teleport it into the rest of distillation circuit using the teleportation circuit of Fig.~(\ref{teleport1}). The distillation would then proceed accounting for discretisation effects as described in Sec.~(\ref{Appendix_disc}). However, in order for teleportation to succeed, after the logical measurement of the first qubit a logical correction on the second needs to be applied
\be
R_z(\phi) X R_z^{\dagger}(\phi) \propto R_z(2\phi) X,
\ee
where proportionality captures an irrelevant global phase.

In quantum computing, commonly $\phi = \pi/2^n$ and $R_z(\phi)$ belongs to the $n$-th level of the Clifford hierarchy. Then, $R_z(2\phi)$ belongs to the $(n-1)$--th level and thus injecting, distilling and teleporting more  magic states to implement the corrections is a terminating process, with number of steps depending on $n$ (see Refs.~\cite{landahl2013complex,campbell2016efficient} for more details). 
For metrology $\phi \in [0,\pi]$ and therefore a similar procedure is not guaranteed to terminate. This, on its own, is not a major issue since we could postselect on measuring $0$ after a $k$ consecutive teleportations with the probability of $1$ being exponentially small on $k$ (teleportation measurements are unbiased). The problem is that distilling a $R_z(2^k \phi)$ rotation, for unknown $\phi$, means interrogating the field with the same state $2^k$ times which will introduce noise of strength $2^k p$. Even for $k=2$, the thresholds we have calculated for the field noise (Protocol Ib) will be worse than the non-FT case (Protocol Ia). Thus, the unknown nature of the rotation, which necessitates using the same field multiple times for the teleportation corrections, means that gate distillation is not giving an benefit over the non-FT protocol.

One could avoid any correction by applying post-selection on the very first teleportation step. This leads the failure probability of one $R_z(\phi)$ application in the distillation circuit to be $1/2$. Since the distillation circuit uses QRMC of block sizes $2^{j+2}-1$ the failure probability of  transversal application on $R_z(\phi)$ on the block is $1-(1/2)^{2^{j+2}-1}$. For $2^{j-1}$ interrogations this amounts to $1-(1/2)^{2^{j-1}(2^{j+2}-1)},$  adding an extra double exponential term in the resource count $C(j)$ from the code. This would be prohibitive.

\paragraph*{Code switching:} We thus resort to the alternative technique of code switching~\cite{1367-2630-17-8-083002,anderson2014fault}. Here, the state is encoded throughout the protocol with a quantum code but not the same at every stage. Code switching exploits the fact that different members of QRMCs are transversal for different gate sets and one can switch between those codes using ancilla qubits and FT measurements. In Protocol Ic we start with a state $\ket{0}$ encoded (by means of FT measurements) by the Steane code and fault tolerantly apply a Hadamard gate in order to prepare the $\ket{+_L}$ probe state. Then we switch to the QRM($1,m$) for $m=j+2$ on which we apply the rest of the protocol.

The circuit applied for each interrogation, Fig.~(\ref{protocolIccircuit}), is similar to that of Protocol Ib (Fig.~(\ref{protocolIbcircuit})). The difference is that the input state $\ket{\psi}$ is already encoded with the required QRMC and therefore the non-transversal operation $E_{\text{QRM}}$ is not needed.  The state is entangled by means of a transversal CNOT gate with the ancilla qubit which is also fault-tolerantly encoded with the same QRM($1,m$) code. 
At every step we apply FT syndrome measurements and recovery operations in the same fashion it is applied in quantum computing \cite{nielsen2010quantum}, the failure probability of which is given in Sec.~(\ref{appendix_ec}). The overhead that comes from the QRM encoding and switching is not counted since we count as resource the number of uses of the field, which are the same as in Protocol Ib.

\subsection{Protocol Ic: Noisy device thresholds}
\label{Appendix_noisydevIc}

Similarly to Sec.~(\ref{Appendix_noisydevIb}), we calculate how the noise in devices affects the error thresholds of Protocol Ic. There are two differences from Protocol Ib. First, the encoding procedure for the QRMCs is now done during the preparation of the probe and ancillae and is fault tolerant. Second, after every operation a round of fault tolerant error correction is applied. The failure probability of the error correction procedure is denoted by $p_{\mathrm{EC}}$ and given in Sec.~(\ref{appendix_ec}). The failure probabilities of the detection procedure are denoted $p''^{X}_{\text{err}}$ and   $p''^{Z}_{\text{err}}$ and given by replacing  probability $p$ by $p+\text{devIc}(p')$ in Eqns.~(\ref{eqn:x_detection}) and~(\ref{eqn:z_detection}) respectively, where
\begin{equation}
\text{devIc}(p')  = (3 \times 2 m + 1 + (2^{m} - m -2) + 1) p'.
\end{equation}
This includes the errors on one qubit from previous syndrome measurements and recovery plus the errors in the error detection syndrome measurements. For our protocol, we need to set $m=j+2$ in the above equation.

Now, the number of FT measurement and recovery steps are $3 \times 2^{j-1} + j + 1$. This includes FT probe preparation, FT Hadamard and $j-1$ steps of switching to QRM$(1,j+2)$, as well as FT ancilla preparation and FT CNOT (two steps)  at each interrogation step. We conservatively approximate the success probability of FT probe preparation, FT Hadamard and each FT switching step by the success probability of FT measurement and recovery step of QRM($1,j+2$). Then the threshold equation becomes
\be
1-(1-p''^{X}_{\text{err}})^{2^{j-1}}(1-p''^{Z}_{\text{err}})^{2^{j-1}}(1-p_{\mathrm{EC}})^{3 \times 2^{j-1}+j+1} = \delta(\gamma').
\label{eq:noisydev3}
\ee
The solution involves two variables and is depicted in Fig.~(\ref{fig:threshold_relations}). We observe that the range of values of $p'_{\text{th}}$ in which $p_{\text{th}}$ is improved over Protocol Ia is smaller than in Protocol Ib with device noise, but, within this region, there is a sub-region where Procotol Ic gives higher thresholds than Protocol Ib. This improvement however is small and the reason for this is the large amount of operations involved in QRMCs error correction.

\subsection{Failure probabilities of QRMCs as error-correcting codes}
\label{appendix_ec}

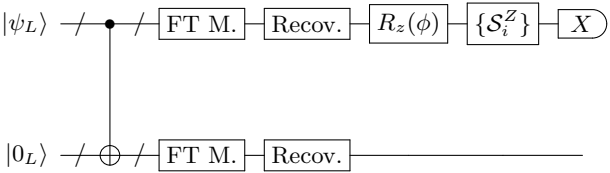
\begin{figure}[t]	
	\mbox{
\Qcircuit  @C=0.8em @R=3.9em {
		& \lstick{\ket{\psi_L}}&  {/} \qw	& \ctrl{1}		& {/} \qw	&\gate{\text{FT M.}}	&  \gate{\text{Recov.}} & \gate{R_z(\phi)}	&  \gate{\{ \mathcal{S}_i^{Z} \}}	&  \measureD{X}		& 		& 	 \\
			& \lstick{\ket{0_L}}& {/} \qw		&\targ	& {/} \qw		&  \gate{\text{FT M.}}	&  \gate{\text{Recov.}} &  \qw   & \qw &  \qw    }
			}
		\caption{Protocol Ic circuit. Operations CNOT, $R_z(\phi)$ and  $X$ measurement are all transversal. Operations $\{ \mathcal{S}_i^{Z} \}$ represent all $Z$ stabilizer measurements. Fault tolerant measurements and recovery require extra ancillae and correct up to $1$ error.}\label{protocolIccircuit}
\end{figure}

To analyse the thresholds of Protocol Ic we calculate the failure probability of the error correction procedure using QRMCs.

Since QRMCs can correct one error of any type, the noise threshold comes from the probability of having two or more errors during all possible operations between two rounds of fault tolerance.
The approximate thresholds for QRM($1,3$) are provided in Ref.~\cite{nielsen2010quantum}. We follow the same techniques to calculate approximate thresholds for QRM($1,m$) for $m>3$.

We begin by enumerating the combinations leading to two errors at the output. We consider the FT measurement and recovery operation on the first logical qubit immediately after the application of transversal CNOT in Fig.~(\ref{protocolIccircuit}).  
The number of ways two errors can occur at the output of the first logical qubit are listed below.
\begin{enumerate}[(i)]
\item Two errors at the previous syndrome measurement and recovery operations. Since there are two blocks with $c_0 = 3 \times 2 \times m \times 2^{m-1} + 2^{m}-1$ points of failure in each, this number is $c_0^2$. 
\item One error at the previous syndrome measurement and recovery operations at one of the two blocks, and another during the logical two qubit gate. This number is $2 c_0 (2^m -1)$.
\item Both during the logical two qubit gate. This number is ${2 (2^m-1) \choose 2}$.
\item Two errors due to incorrect syndrome measurement. This number is $( 2 m) {2 \times 2^{m-1}) \choose 2}$. 
\item Both at the syndrome measurements: $c_0^2$.
\item One at the syndrome measurement and another during recovery: $c_0 (2^m-1)$.
\item Both during recovery: $(2^m-1)^2$.
\end{enumerate}
Summing all the above contributions, we get
\begin{eqnarray}
c &=& 2 c_0^2  + {{2 (2^m - 1)}\choose{2}} +(2 m) {2^m \choose 2} +  3c_0(2^m - 1 )    \nonumber \\ 
 && + (2^m - 1)^2.  \nonumber
\end{eqnarray}
The  probability of failure of the error correction procedure, which is the probability of having at least two errors is then
\begin{equation}
p_{\mathrm{EC}} \approx c p'^2,
\end{equation}
where $p'$ is the probability of a single component in the device being affected by noise.  

\section{Parallel protocols}
\label{appendix_parallel}

A parallel version of our protocols Ia, Ib, and Ic can be implemented by preparing GHZ states of $2^{j-1}$ entangled qubits and interrogating the field in parallel, as depicted in Fig.~(\ref{withoutFT}).
The performance of this parallel version for Protocols Ia, Ib without device noise is identical to the serial versions. 
The contribution of noisy devices in Protocol Ib and Ic is different due to a different preparation and measurement procedures compared to the serial protocol.
Since both the serial and the parallel version require the application of Hadamard gate, we need in both cases to complement the QRM($1,j+2$) with a code transversal for the $H$ gate.

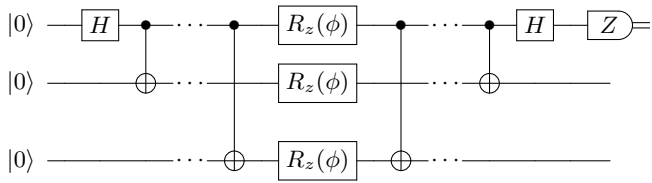
\begin{figure}[h!]
\centering
\Qcircuit  @C=0.7em @R=0.8em {
	&	&	& \lstick{\ket{0}}& \qw	& \gate{H} & \ctrl{1} & \qw &  {\cdots} & &  \ctrl{3}	& \qw & \gate{R_z(\phi)} &\qw & \ctrl{3} & \qw &  {\cdots} & & \ctrl{1} &  \gate{H} & \qw &  \measureD{Z}		& \cw		&  \\
	&	&	& \lstick{\ket{0}}& \qw	& \qw	& \targ	& \qw &  {\cdots} & & \qw	&\qw & \gate{R_z(\phi)} &\qw & \qw	& \qw &  {\cdots} & & \targ &  \qw		&  \qw & \qw \\
	\\
	&	&	& \lstick{\ket{0
	}}& \qw	& \qw	& \qw	& \qw &  {\cdots} & & \targ	&\qw & \gate{R_z(\phi)} &\qw & \targ & \qw &  {\cdots} & & \qw		&  \qw & \qw &  \qw	}
\caption{Parallel phase estimation without fault tolerance}
\label{withoutFT}
\end{figure}

\section{discussions}
\label{sec:disc}

We have illustrated a methodology for FT quantum metrology that allows estimation of phase up to higher bits of precision in the presence of arbitrary local Pauli noise. This is based on improved noise thresholds for our phase estimation scheme.   
While we have focussed on the principle of FT quantum phase estimation, its practical use will depend on reducing resource consumption and increasing thresholds improvements. This should direct future work by calculating fault tolerance thresholds and resources for other known schemes, both non-adaptive~\cite{kitaev1995quantum, higgins2009demonstrating} as well as adaptive~\cite{griffiths1996semiclassical,
cleve1998quantum,nielsen2010quantum,
higgins2009demonstrating}.

Noise thresholds have been identified~\cite{kimmel2015robust} for non-adaptive phase estimation schemes~\cite{higgins2009demonstrating} under general additive noise and establish a noise threshold for a modified version of it. While these works do not use QECC or fault tolerance, they do possess thresholds better than ours. This is due to the sophistication of the estimation scheme, and its fault tolerance would therefore be an interesting open question.

Further improvements in FT quantum metrology should be possible with better estimation schemes as well as the quantum error correcting codes, the latter determined by the transversality demands set by the unknown parameter(s) to be estimated. These should spur developments not only in quantum metrology but also quantum error correction and fault tolerance. 

\section{Acknowledgements}

We thank T. Rudolph for communications about Ref.~\cite{rudolph2003quantum}, B. Terhal for pointing to Ref.~\cite{kimmel2015robust}, D. Branford for technical discussions, J. Friel for commenting on the manuscript and S. Ferracin for graphics assistance. This work was partly supported by the UK EPSRC (EP/K04057X/2), and the UK National Quantum Technologies Programme (EP/M013243/1,EP/M01326X/1).

\appendix

\section{Quantum error correction}
\label{appendix_ft}

It is known~\cite{jochym2017disjointness} that transversal gates on stabilizer codes are necessarily at a finite level of the Clifford hierarchy~\footnote{Let $\mathcal{P}$ be the group of Pauli operators. The first level of the Clifford hierarchy $\mathcal{C}_1$ is the normalizer of the Pauli group under conjugation. Then, the $n$-th level of the Clifford hierarchy $\mathcal{C}_n$ is the set of operators that map the Pauli group to the $(n-1)$-th level of the hierarchy under conjugation. A rotation operator $\text{diag}(1,e^{i 2 \pi/2^n})$ belongs to the $(n-1)$-th level of the Clifford hierarchy. Thus, a transversal rotation by a real angle can potentially corrupt the logical space of a stabilizer code. More in Appendix~\ref{appendix_ft} and Ref.~\cite{jochym2017disjointness}}. This is based on the notion of disjointness, which is a metric of stabilizer quantum error-correcting codes and  is, roughly speaking, the number of mostly non-overlapping representatives of any given non-trivial logical Pauli operator.

\begin{theorem} [Theorem $5$ in \cite{jochym2017disjointness}]
Consider a stabilizer code with min-distance $d_\downarrow$, max-distance $d_\uparrow$ and disjointness $\Delta$. If $M$ is an integer satisfying
$$d_\uparrow < d_\downarrow \Delta^{M-1},$$
then all transversal logical operators are in the $M^{\text{th}}$ level of the Clifford hierarchy $\mathcal{C}_M$.
\end{theorem}

This theorem implies that in our construction for FT metrology, we cannot hope to use a stabilizer code that is transversal for any gate $R_z(\phi)$ for $\phi \in \mathbb{R}$.

\emph{Reed-Muller codes:} 
Reed-Muller codes RM($r,m$) of block length $n=2^m$, for $0 \leq r \leq m$, dimension $\sum_{i=0}^{r} {{m} \choose {i}}$ and distance $2^{m-r}$ are a family of classical block codes~\cite{macwilliams1977theory}. Reed-Muller codes have geometric properties that allow for easy decoding. Codewords of RM($r,m$) correspond to all Boolean functions $f$ of $m$ variables of degree $r$. Each codeword is the last column of the truth table of $f$, i.e. the values of $f$ for all different inputs. For example, the rows of the generator matrix of RM($1,3$) contain the values of $a_0 \boldsymbol{1} + a_1 x_1 + a_2 x_2 + a_3 x_3$, where $\boldsymbol{1}$ stands for the vector of all ones, for all $x_i$'s and each row corresponds to a different element of a basis on $a_i$'s:
$$ G = \begin{bmatrix}
   0  & 0   & 0 & 0 & 1 & 1 & 1 & 1 \\
    0  & 0   & 1 & 1 & 0 & 0 & 1 & 1  \\
    0  & 1   & 0 & 1 & 0 & 1 & 0 & 1  \\
    1  & 1   & 1 & 1 & 1 & 1 & 1 & 1 
\end{bmatrix}.$$

We are interested in the divisibility properties of Reed-Muller codes, which has implications for the transversality of the quantum Reed-Muller codes.  A classical code $C$ is called divisible by $\Delta$ if $\Delta$ divides the weight of all $x \in C$. A code is called divisible if it is divisible by $\Delta >1$.  First order RM codes are divisible by $2^{m-1}$ because exactly half of the outputs of a boolean function of degree $1$ have value $1$, except function $\boldsymbol{1}$ which always has output $1$.

QRMCs use codes constructed from RM codes. We present their divisibility properties and weight distribution. The shortened RM code, denoted by $\overline{RM}$, is taken by keeping only the codewords which begin with $0$ and delete their first coordinate. Codewords of $\overline{RM}(1,m)$  can be defined by the following recursive process.
For $m=2$
$$\mathcal{S}_{2} = \begin{bmatrix} 0 & 0 & 0 \\ 0 & 1 & 1 \\ 1 & 0 & 1 \\ 1 & 1 & 0   \end{bmatrix},$$ 
and for higher values of $m$ 
$$\mathcal{S}_m = \begin{bmatrix}  \mathcal{S}_{m-1} & \boldsymbol{0} & \mathcal{S}_{m-1} \\ \mathcal{S}_{m-1} & \boldsymbol{1} & \overline{\mathcal{S}}_{m-1} \end{bmatrix}.$$
Code $\overline{RM}(1,m)$ therefore has one codeword of weight $0$ and $2^m-1$ of weight $2^{m-1}$.

The punctured Reed-Muller code $RM^*$ is obtained by adding the $\boldsymbol{1}$ row to the generator of $\overline{RM}$. 
$RM^*$ therefore has one codeword of weight $0$, $2^m-1$ of weight $2^{m-1}-1$, $2^m-1$ of weight $2^{m-1}$ and one of weight $2^m-1.$

\emph{Quantum Reed-Muller codes:} 
A quantum Reed-Muller code QRM($1$,$m$) is a CSS code based on classical Reed-Muller codes. It is constructed using the punctured Reed-Muller code $RM^*$ and its even subcode $\overline{RM}$ with logical states
$\ket{x}_L \equiv \sum_{\boldsymbol{y} \in \overline{RM}} \ket{\boldsymbol{y} + x \boldsymbol{1}},$ for $x \in \{0,1\}$. The size of the block is $2^m-1$ qubits. The minimum distance is $3$, which is the minimum distance of the dual of the $\overline{RM}$ that is used to correct the $Z$ errors~\footnote{$RM^*$ code used to detect the $X$ errors and the dual of $\overline{RM}$ code used to detect the $Z$ errors have different distances, a fact exploited by our scheme.}.

Using the following Lemma, we justify our choice of QRM($r,m$) code with $r=1$ and $m$ chosen according to transversality requirements. 

\begin{lemma}[Corollary 4 in \cite{zeng2011transversality}]
Let QRM($r,m$) created by the construction described above, where $0<r\leq \lfloor m/2 \rfloor$. Then it is an $[n=2^m -1, 1, d= \min (2^{m-r}-1, 2^{r+1} -1)]$ code, with transversal $T_t$ for $t=\lfloor m/r \rfloor -1$. 
\end{lemma}

We can thus calculate the failure probability for QRM($r,m$) with $r=1$ and for $r>1$ with a fixed $m/r$ ratio, to have the same transversality property, with an error model where each physical qubit is corrupted with probability $0 \leq p \ll 1$. We calculated the thresholds for $r=2$, using the following theorem:

\begin{theorem}[Theorem $8$, Ch. $15$,  Ref.~\cite{macwilliams1977theory}]
Let $A_i$ be the number of codewords of weight $i$ in RM($2,m$). Then $A_i = 0$ unless $i=2^{m-1}$ or $i=2^{m-1} \pm 2^{m-1-h}$ for some $h$, $0 \leq h \leq \lceil \frac{m}{2} \rceil$. Also, $A_0 = A_{2^m} =1 $ and $A_{2^{m-1} \pm 2^{m-1-h}} = 2^{h(h+1)} \frac{(2^m-1)(2^{m-1}-1)\ldots(2^{m-2h+1}-1)}{(2^{2h}-1)(2^{2h-2}-1)\ldots(2^2-1)}$.
Finally, $A_{2^{m-1}}=2^{1+m+{m \choose 2}} - \sum_{i \neq 2^{m-1}} A_i$.

\end{theorem}

We use the technique of Sec.~(\ref{Appendix_thres}) to calculate the thresholds. We find the thresholds for $r=2$ with the same transversal properties worse than the case $r=1$. 
Given the threshold calculations for $\geq 3$ are too prolix and the block sizes too large, we choose $r=1.$ 

Transversality of QRM($1$,$m$) is based on the fact that all the codewords of $\overline{RM}$ are divisible by $\Delta=2^{m-1}$, while their complement is divisible by $\Delta=2^{m-1}-1$ or $\Delta=2^{m}-1$. 
Transversality enables different operations on each logical computational basis state modulo $2^{m-1}$, by applying transversal gates on the $2^{m}-1$ physical qubits. In particular applying transversal $T_n$ on QRM($1,n+1)$ will apply the logical $T_n^{\dagger}$ gate. For example QRM($1,4)$ is transversal for $T_3=T$, also known as the phase-$\pi/8$ gate, but not for smaller fractions of rotations around the $Z$-axis.


\section{Mixed radix extension of RG estimator}
\label{appendix_ji}

The RG estimator~\cite{rudolph2003quantum} only converges if the phase $\phi$ lies in certain regions as detailed in Sec.~(\ref{AppendixRG}).
This limitation can be overcome by slightly modifying the estimator~\cite{ji2008parameter}. It beings by expressing the parameter in a mixed base as 
\be
\phi=v_1 \frac{\pi}{r_1} + v_2 \frac{\pi}{r_1 r_2} + v_3 \frac{\pi}{r_1 r_2 r_3} \ldots,
\ee
where $r_i \in \{2,3\}$.
In order to estimate dit $j$ the qubit $\ket{+}$ state interrogates the field an appropriate number of times followed by a Pauli $X$ measurement. The protocol is identical to that depicted in Fig.~(\ref{fig:protocols}), only with a different number of consecutive interrogations.
Unlike Protocol Ia, Protocol II below converges for all values of $\phi$, since there is no excluded region (Fig.~(\ref{fig:estimator2})).

\begin{figure}[t!]
		\includegraphics[scale=0.2]{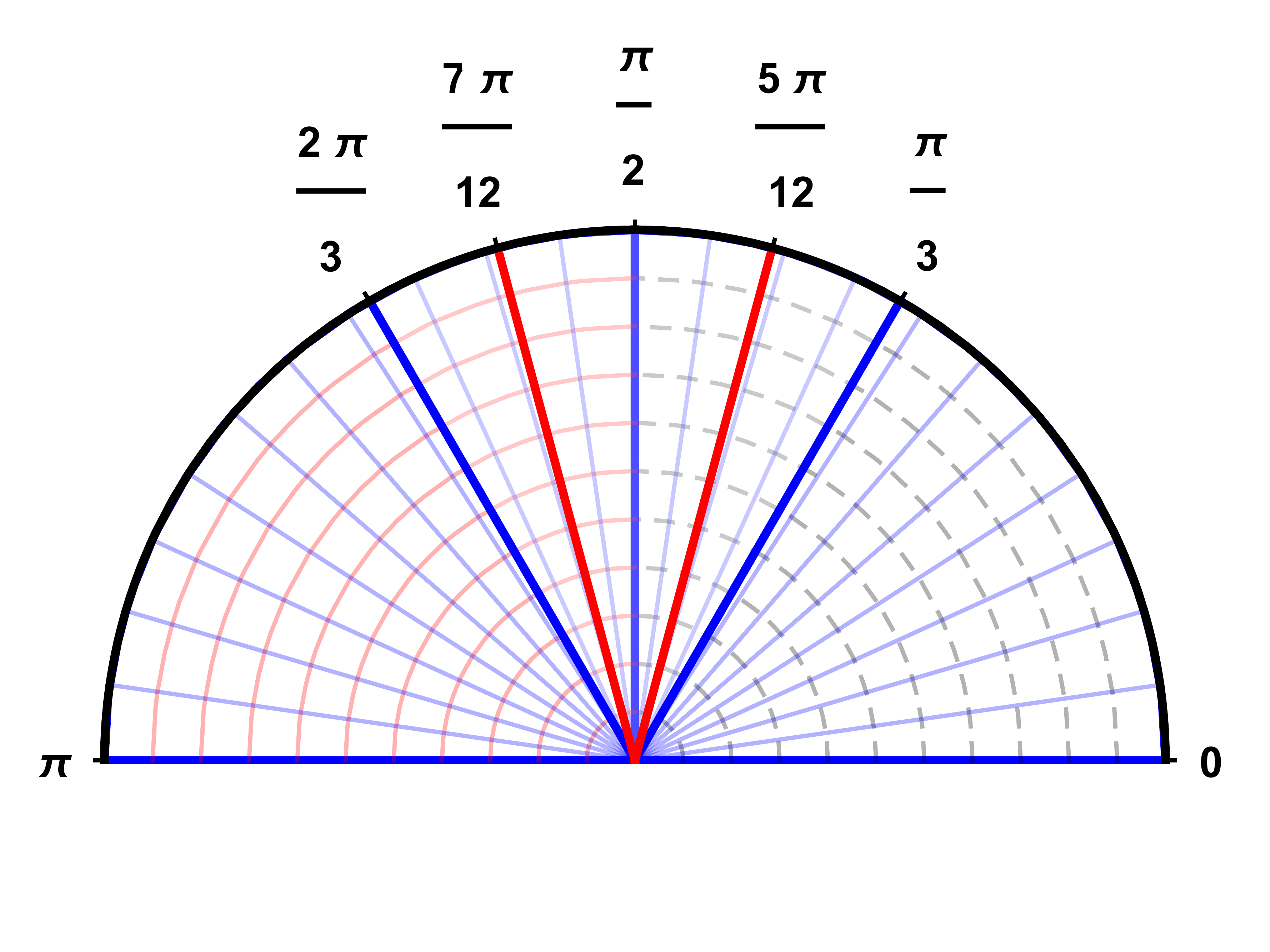}
		\caption{Estimator of Protocol II. There are three regions: $[0,\pi/2]$, $[\pi/3,2\pi/3]$ and $[\pi/2,\pi]$. The decision boundaries for $j=1$ are the red lines. There are no excluded regions.}\label{fig:estimator2}
\end{figure}

\hrulefill

\textbf{Protocol II -- Extended RG estimator~\cite{ji2008parameter}}

For $j=1, \ldots, t$

\begin{enumerate}
\item Repeat $M$ times\\
	(i) Prepare $\ket{+}$. \\
 	(ii) Interrogate field $\prod_{l=0}^{j-1} r_l$ times ($r_0 =1$).\\
	(iii) Measure $X$.
\item Calculate $\widehat{p}_j$ as the fraction of the $+1$ measurement outcomes out of $M$. 
If $\widehat{v}_{j-1}=0$ set $\widehat{\phi}_j = \cos^{-1} (2 \widehat{p}_j -1)$ in $[0,\pi]$ , or else in $[\pi,2\pi]$. If\\
(i) $\widehat{v}_{j-1} \pi \leq \widehat{\phi}_j <  \widehat{v}_{j-1} \pi + \frac{5\pi}{12} $, set $\widehat{v}_j=0$ and $r_j=2$.\\
(ii) $\widehat{v}_{j-1} \pi + \frac{5\pi}{12}  \leq \widehat{\phi}_j < \widehat{v}_{j-1} \pi + \frac{7\pi}{12} $, set $\widehat{v}_j=1$ and $r_j=3$.\\
(iii) $\widehat{v}_{j-1} \pi + \frac{7\pi}{12} \leq \widehat{\phi}_j \leq \widehat{v}_{j-1} \pi + \pi$, set $\widehat{v}_j=1$ and $r_j=2$.

\item If $j \neq t$ add $1$ to $j$ and go to step $1$, otherwise exit and output
 $$\widehat{\phi}= \widehat{v}_1 \frac{\pi}{r_1} + \widehat{v}_2 \frac{\pi}{r_1 r_2} \ldots$$
\end{enumerate}

\vskip-0.12in

\hrulefill

The convergence of the noiseless protocol is proven in \cite{ji2008parameter}. Here we discuss its noise resilience following the analysis for the noisy Protocol Ia in Sec.~(\ref{AppendixRG}). 

Following Protocol Ia, $\gamma$ is the maximum error allowed in the estimated angle for the protocol to converge. In Protocol II, $\gamma$ is fixed to $\pi/12$ since an estimation within this error means that if\\
 (i) $0\leq \phi_j - v_{j-1} \pi < \pi/2 \Rightarrow (\widehat{v}_{j} = v_j = 0|r_j=2) $;\\
 (ii) $\pi/3 \leq \phi_j - v_{j-1} \pi < 2\pi/3 \Rightarrow (\widehat{v}_{j} = v_j = 1|r_j=3)$;\\
 (iii) $\pi/2 \leq \phi_j - v_{j-1} \pi < \pi \Rightarrow (\widehat{v}_{j} = v_j = 1|r_j=2)$.

The associated maximum error in the estimated probability is $\delta = \left|\cos^2 \left( \frac{5\pi}{24} \right) - \cos^2 \left( \frac{6\pi}{24} \right) \right| \approx 0.129. $
The noise thresholds are given by the solutions to
\be
1 - (1-p_{\text{th}})^{3^{t-1}} = \delta
\ee
since $1 - (1-p)^{\prod_l^{t-1} r_l} \leq 1- (1-p)^{3^{t-1}}$.
The number of field interrogations, our resource, to estimate $t$ dits with error $\epsilon$ is
\be
N = \sum_{j=1}^t \prod_l^{j-1} r_l \frac{1}{2(\delta - p_f)^2} \ln \left(\frac{2t}{\epsilon}\right).
\ee
A FT protocol for this estimator on the lines of Protocol Ib using QRMCs suffers from non-transversality for phases such as $\phi=\pi/3.$ A logical shift in such a phase pushes logical angles $\phi'_j$ outside $[0,\pi]$ because $3$ times the logical rotation corresponding to transversal $\pi/3$ does not equal $\pi$. This induces an error in the estimation.

A convergent FT protocol is therefore impossible if we restrict ourselves to codes transversal for rotations $\pi/2^k$ unless we can interrogate the field for a fractional amount of time depending on $j$ and the corresponding logical phase shift given by Lemma~(\ref{lemma_disc}).

\onecolumngrid


\section{Graphs for different values of $\gamma$}
\label{appendix:extra_graphs}

\begin{figure*}[h]	
	\begin{subfigure}[h!]{0.31\textwidth}
		\includegraphics[scale=0.44]{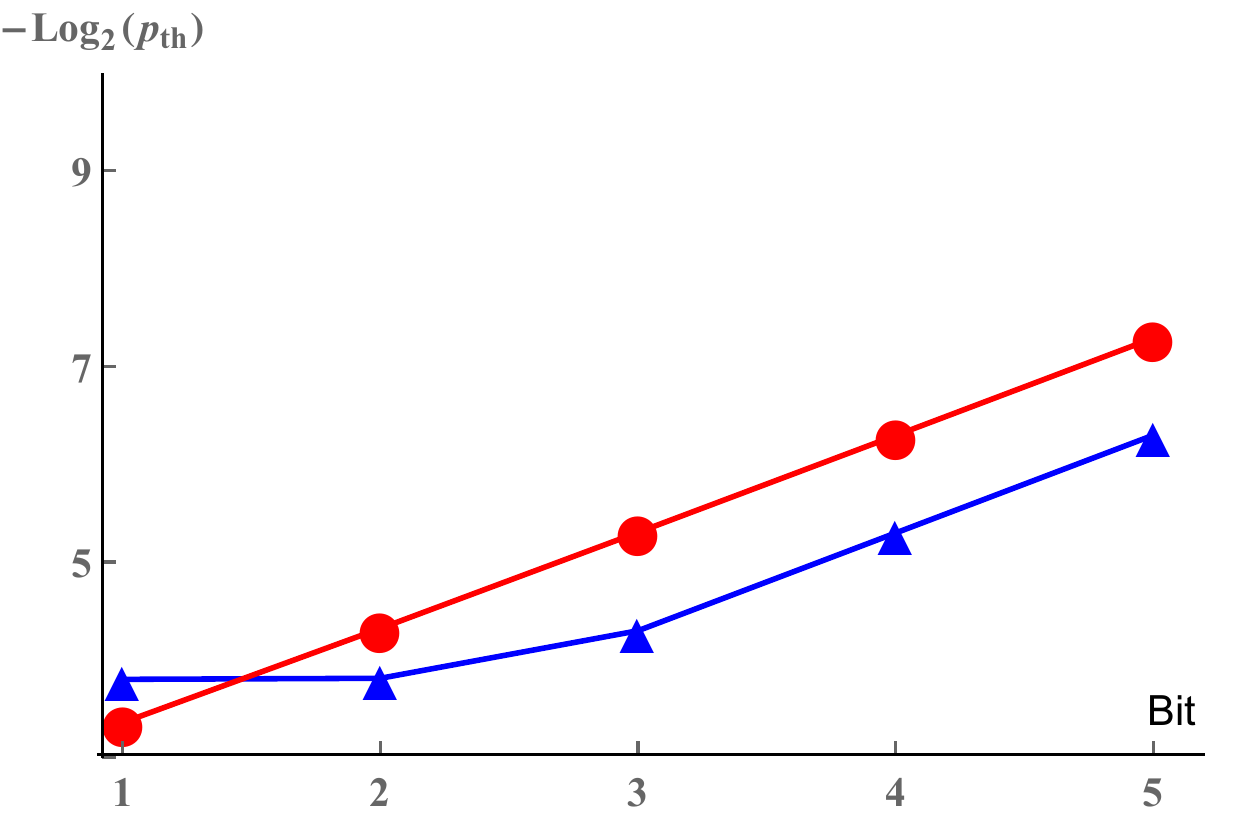}
		\caption{$\gamma=\pi/16$}\label{fig:threshhh1}
	\end{subfigure}
	\quad 
	\begin{subfigure}[h!]{0.31\textwidth}
		\includegraphics[scale=0.44]{threshold_main_paper}
		\caption{$\gamma=\pi/32$}\label{fig:threshhh2}
	\end{subfigure}
		\quad 
	\begin{subfigure}[h!]{0.31\textwidth}
		\includegraphics[scale=0.44]{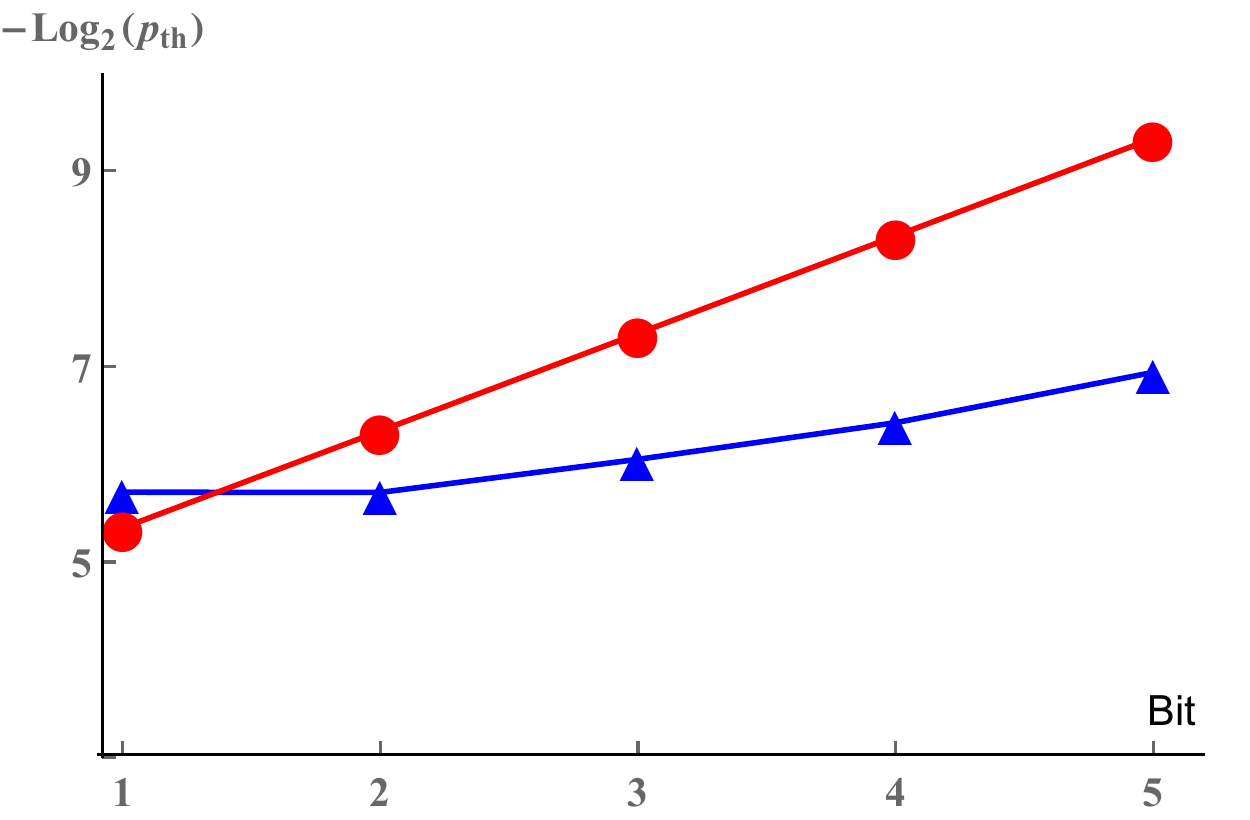}
		\caption{$\gamma=\pi/64$}\label{fig:threshhh3}
	\end{subfigure}
	\caption{Interrogation noise thresholds. Red: Protocol Ia; Blue: Protocol Ib without device noise. For the effect of device noise see Fig.~(\ref{fig:threshold_relations}).}
\label{fig:threshhh}
\end{figure*}

\begin{figure*}	
	\begin{subfigure}[h!]{0.31\textwidth}
		\includegraphics[scale=0.44]{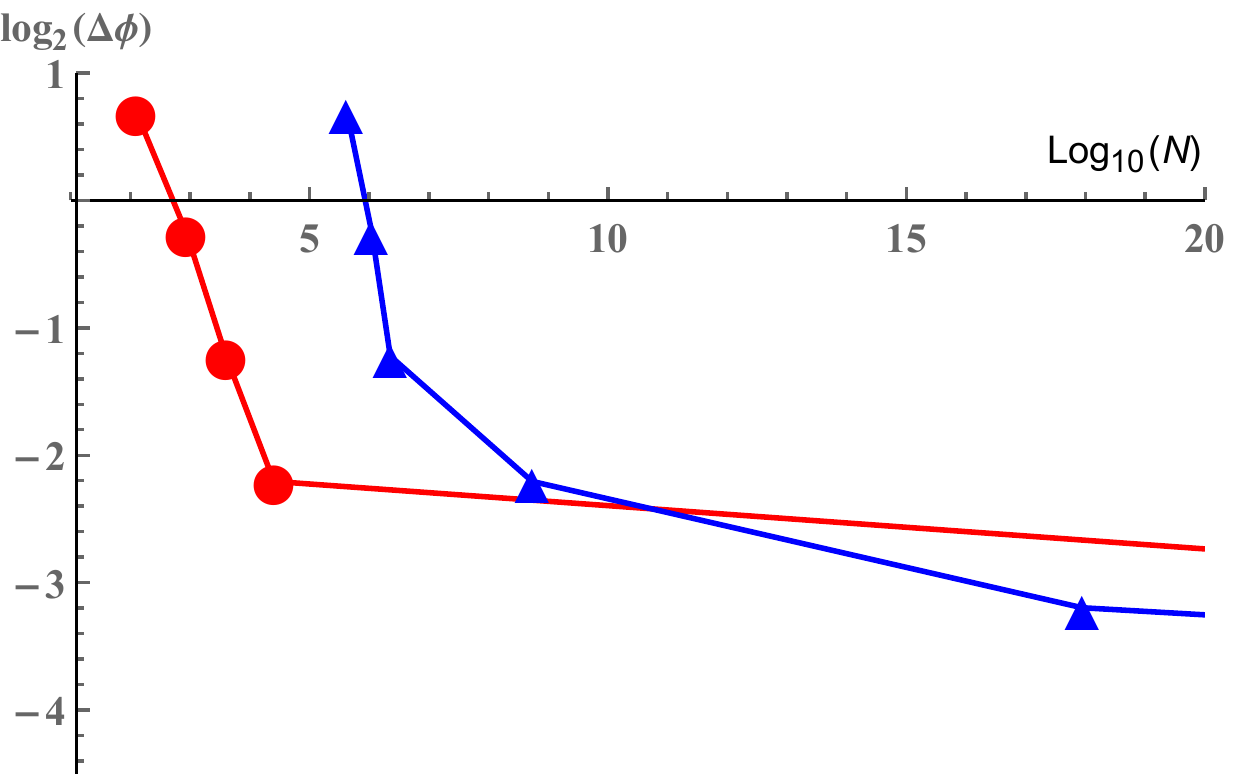}
		\caption{$\gamma=\pi/16$}\label{fig:varresourcc1}
	\end{subfigure}
	\quad 
	\begin{subfigure}[h!]{0.31\textwidth}
		\includegraphics[scale=0.44]{resources_main_paper_}
		\caption{$\gamma=\pi/32$}\label{fig:varresourcc2}
	\end{subfigure}
		\quad 
	\begin{subfigure}[h!]{0.31\textwidth}
		\includegraphics[scale=0.44]{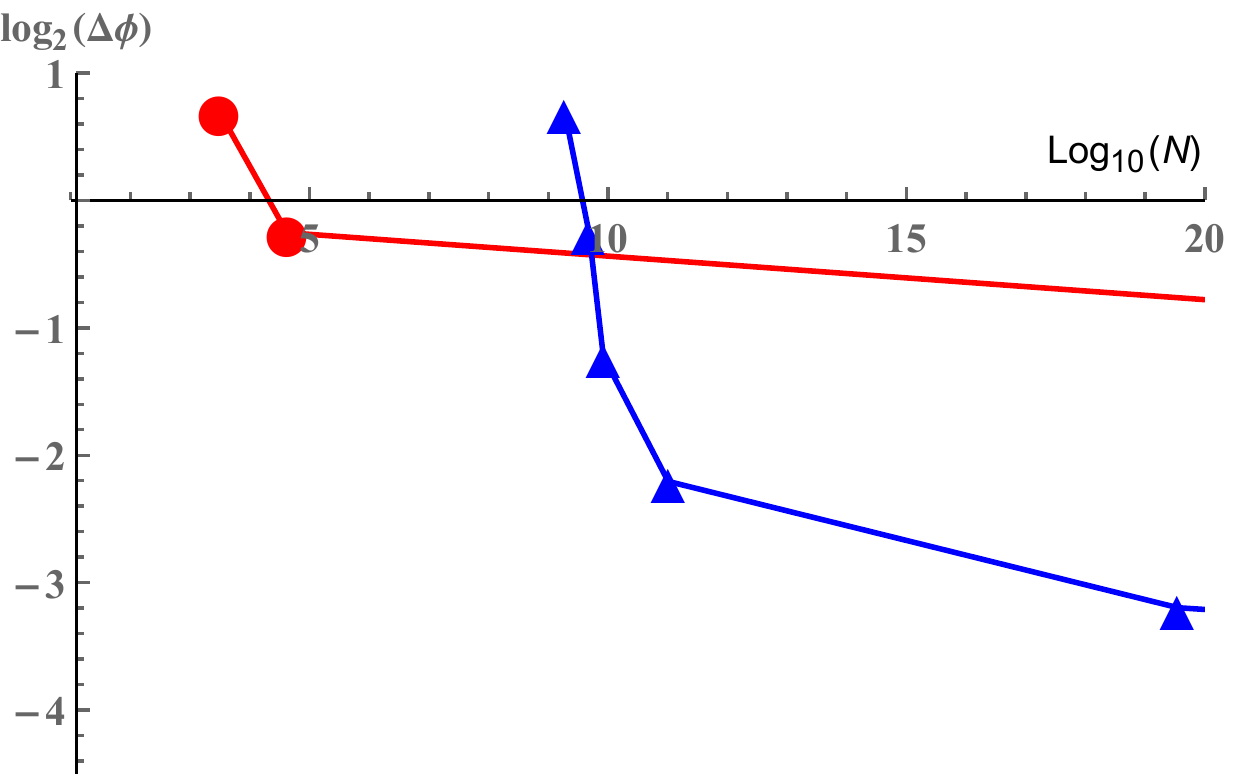}
		\caption{$\gamma=\pi/64$}\label{fig:varresourcc3}
	\end{subfigure}
	\caption{Precision as a function of number of interrogations, our resource. Red: Protocol Ia. Blue: Protocol Ib without device noise, and markers denote bits. Improvement from fault tolerance is illustrated in estimating higher bits. Noise chosen at the noise thresholds of Protocol Ia which are closer to $0.63\%$: (a) $p=0.639\%$, (b) $p=0.626\%$ and (c) $p=0.619\%$.}
\label{fig:varresourcc}
\end{figure*}

\twocolumngrid

\bibliography{MetrologyReferences}

\end{document}